\documentclass[aps,prl,reprint,superscriptaddress,nofootinbib,floatfix]{revtex4-2}

\usepackage[T1]{fontenc}
\usepackage[utf8]{inputenc}
\usepackage{lmodern}
\usepackage{microtype}
\usepackage{amsmath,amssymb,amsthm,mathtools,bm}
\usepackage{dsfont}
\usepackage{booktabs,tabularx,array}
\usepackage{xcolor}
\usepackage{graphicx}
\usepackage{hyperref}
\hypersetup{
  colorlinks=true,
  linkcolor=blue!45!black,
  citecolor=blue!45!black,
  urlcolor=blue!45!black,
  pdftitle={Experimentally Testable Quantum Advantage in Shallow Circuits},
  pdfauthor={Kishor Bharti and Adan Cabello}
}

\newtheorem{theorem}{Theorem}
\newtheorem{lemma}[theorem]{Lemma}
\newtheorem{proposition}[theorem]{Proposition}
\newtheorem{corollary}[theorem]{Corollary}
\theoremstyle{definition}
\newtheorem{definition}[theorem]{Definition}
\newtheorem{construction}[theorem]{Construction}
\theoremstyle{remark}
\newtheorem{remark}[theorem]{Remark}

\newcommand{\cval}{\omega_{\mathrm c}}
\newcommand{\qval}{\omega_{\mathrm q}}
\providecommand{\Game}{}
\renewcommand{\Game}{\mathsf G}
\newcommand{\Comp}{\mathsf{Comp}}
\newcommand{\Prb}{\operatorname{Pr}}
\newcommand{\Tr}{\operatorname{tr}}
\newcommand{\NLF}{\operatorname{NLF}}
\newcommand{\LC}{\operatorname{LC}^{-}}
\newcommand{\dist}{\operatorname{dist}}
\newcommand{\ket}[1]{\lvert #1\rangle}

\newcommand{\id}{\mathds{1}}

\begin{document}

\title{Experimentally Testable Quantum Advantage in Shallow Circuits}

\author{Kishor Bharti}
\altaffiliation{Currently at IonQ Inc.}
\email{kishor.bharti1@gmail.com}
\affiliation{Joint Center for Quantum Information and Computer Science, NIST/University of Maryland, College Park, Maryland 20742, USA}
\affiliation{University of Maryland Institute for Advanced Computer Studies, University of Maryland, College Park, Maryland 20742, USA}

\author{Ad{\'a}n Cabello}
\email{adan@us.es}
\affiliation{Departamento de F{\'i}sica Aplicada II, Universidad de Sevilla, E-41012 Sevilla, Spain}
\affiliation{Instituto Carlos I de F{\'i}sica Te\'orica y Computacional, Universidad de Sevilla, E-41012 Sevilla, Spain}

\begin{abstract}
Experimental tests of shallow-circuit quantum advantage require explicit classical bounds at finite circuit sizes. {We refine the finite-size classical soundness bound of Aasn{\ae}ss's graph-distributed construction to depend linearly on the number of players. Combined with standard disjoint-player repetition, this gives a two-round test on a single processor for any finite nonlocal game with a finite-dimensional perfect quantum strategy and classical winning probability $\gamma<1$, with arbitrarily small classical success.} The method is based on playing $m$ copies with disjoint players and teleporting each player's register to a uniformly random one of $N$ sites before the questions are revealed. Each answer bit of a depth-$D$, fan-in-$K$ classical response with fixed wiring depends on at most $K^D$ question wires, making cross-player dependencies unlikely when $N$ is large. The quantum implementation 
has constant depth per round and wins with certainty. A classical device wins with probability at most $\gamma^m+O(mK^D/N)$, which vanishes as $O(\log N/N)$ for $m=\lceil\log_{1/\gamma}N\rceil$ at fixed $D$ and $K$. We present an explicit proposal for an experimentally testable quantum advantage with 99 qubits.
\end{abstract}

\maketitle

\paragraph{Introduction.}
A claim of quantum advantage must specify the classical computational model being excluded. The quantum advantage of sampling experiments~\cite{Arute2019,Zhong2020,Wu2021,Madsen2022} is conditional on complexity conjectures~\cite{HarrowMontanaro2017,HangleiterEisert2023}. Two alternative approaches instead fix the classical model and prove an upper bound within that model~\cite{CleveHoyerTonerWatrous2004,BGK2018}. Nonlocal games have {well-defined classical values~\cite{Palazuelos2016}} and have become practical benchmarks~\cite{Xu2022,Daniel2022,Sheng2025,Drmota2025,Than2026}, but exceeding the classical value certifies nonclassicality only under the assumption {(which is not enforced by spacelike separation on a single processor)} that the players cannot communicate.
Cryptographic compilations of nonlocal games replace the no-communication assumption by computational assumptions~\cite{KahanamokuMeyer2022,KalaiLombardiVaikuntanathanYang2023,NatarajanZhang2023,KulpeMalavoltaPaddockSchmidtWalter2025}. Shallow-circuit separations, initiated by Bravyi, Gosset, and K\"onig (BGK)~\cite{BGK2018} and developed in many directions since~\cite{LeGall2019,CoudronStarkVidick2021,WattsKothariSchaefferTal2019,OliveiraSubramanianMendesHsieh2025,GrewalKumar2025,GriloKashefiMarkhamOliveira2026,HsiehOliveiraSubramanianZhang2026,GrierSchaeffer2020,GrierJuSchaeffer2021,BGKT2020,HasegawaLeGall2021,BhartiJain2023,CahaCoiteuxRoyKoenig2024,CahaCoiteuxRoyKoenig2026,ZhangPanLiu2024,ZhangGongLiDeng2024,PirnayJerbiSeifertEisert2024,WattsGossetLiuSoleimanifar2025,McGinleyHoMalz2025,WattsParham2026,GrierKaneMorrisOstuniWu2026}, replace the no-communication assumption by a restriction on circuit depth: a constant-depth quantum circuit solves a task that no classical circuit of bounded fan-in and comparable depth can, without any complexity assumption. {Related measurement-driven and adaptive shallow-circuit proposals use sampling or discrete-logarithm hardness assumptions~\cite{CaoEisert2026,WuZhangZhangWangYuan2026}.} 
BGK give a classical success bound for sufficiently large instances~\cite[Theorem~2]{BGK2018}, and finite shallow-circuit benchmarks have been studied experimentally with a nearest-neighbor restriction on classical connectivity~\cite{Daniel2022}.
We seek a finite experimental test with an unconditional advantage against a classical competitor that may be wired arbitrarily, a classical threshold that can be made as small as desired, and an acceptance rule for a finite number of trials.

{We sharpen the finite-size classical soundness bound of Aasn{\ae}ss's graph-distributed construction~\cite{Aasnaess2022} from quadratic to linear dependence on the number of players, and combine it with standard disjoint-player repetition to specify finite experimental tests.} Its input is any finite nonlocal game $\Game$ with a finite-dimensional perfect quantum strategy, $\qval(\Game)=1$, and classical value $\cval(\Game)=\gamma<1$, where $\qval$ and $\cval$ are the optimal quantum and classical winning probabilities.

The test is constructed from $\Game$ in two steps. First, $m$ copies of $\Game$ are played by $m$ disjoint sets of players. Then, the classical value becomes exactly $\gamma^m$ while quantum players still win with certainty. Second, each player's quantum register is teleported through Bell pairs to a uniformly random one of $N$ sites, the questions are revealed only afterwards and delivered to those sites, and the answers are read there. This step specializes the graph-distributed compiler of Aasn{\ae}ss~\cite{Aasnaess2022}, which generalizes BGK to arbitrary games and preserves the success probability of any finite-dimensional strategy [Eq.~(4.46) of Ref.~\cite{Aasnaess2022}]; we use a perfect strategy so that all $m$ copies keep winning with certainty. A classical circuit whose response to the questions has depth $D$ and fan-in $K$ can consult at most $K^D$ input wires per answer bit, so the probability that any answer depends on another player's question is small for sufficiently large $N$. When no such dependency occurs, the answers form a local strategy, which wins with probability at most $\gamma^m$. The resulting bound, presented in Eq.~\eqref{eq:main-bound}, has explicit constants and vanishes as $O(\log N/N)$ for $m=\lceil\log_{1/\gamma}N\rceil$ at fixed $D,K$; classical depth $\Omega(\log N)$ is necessary for any fixed positive success probability, and depth $O(\log N)$ suffices to win with certainty. The quantum device starts in a product state and has constant depth in each of two rounds; the certification criterion uses the measured winning frequency; losses count as failures; and the threshold is a number fixed before data are taken. We give a 99-qubit Greenberger--Horne--Zeilinger (GHZ) instance implementing the two-round protocol and multi-copy instances whose thresholds lie well below $3/4$.

\begin{figure*}[t]
\centering
\includegraphics[width=\textwidth]{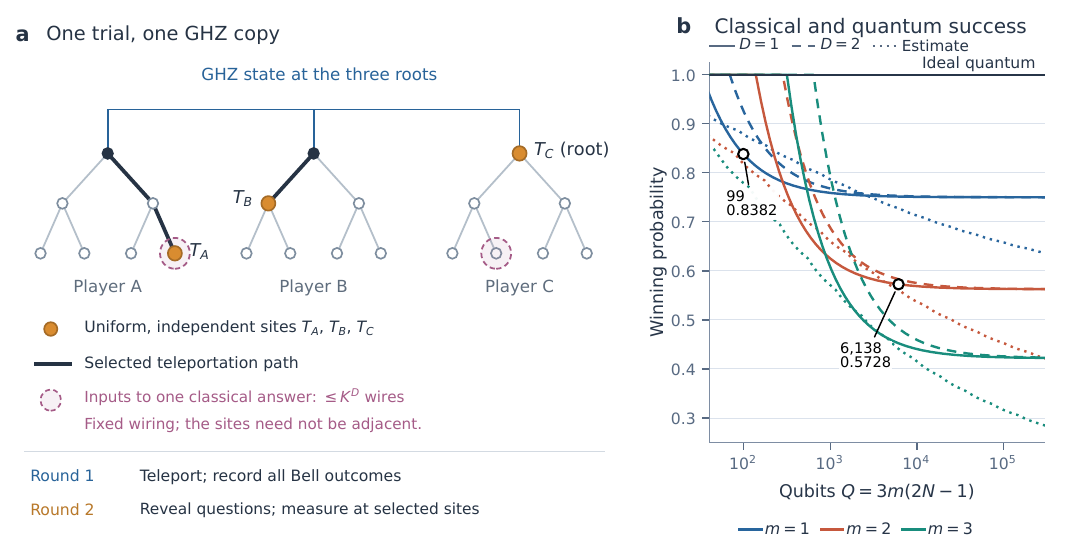}
\caption{(a) One trial for one GHZ copy. Each player's qubit is teleported through Bell pairs along its host graph to a uniformly random site $T_j$ (orange). The questions are revealed after the Bell outcomes are recorded and are delivered only to those sites. An answer from a depth-$D$, fan-in-$K$ classical circuit depends on at most $K^D$ input wires (purple circles illustrate two possible input locations). (b)~Classical threshold $B_{\rm c}$ from Eq.~\eqref{eq:main-bound} against the number of qubits $Q=3m(2N-1)$ for tree hosts, with $m=1,2,3$ GHZ copies and fan-in $K=2$. Solid curves have $D=1$; dashed curves have $D=2$. The ideal quantum winning probability is one. Circles mark the 99-qubit instance and the $m=2$, $N=512$ instance. Dotted curves show the quantum estimate in Eq.~\eqref{eq:sm-hardware-estimate} for binary-tree hosts, with gate and readout error rates from Ref.~\cite{Kumar2025Contextuality} and assumed baseline correlators $V_0=0.95$. The estimate omits additional idling and routing errors.}
\label{fig:main-protocol}
\par\smallskip\begin{minipage}{\textwidth}\small\raggedright 
\end{minipage}
\end{figure*}

\paragraph{Games and the classical competitor.}
{A finite $b_0$-player game $\Game$ consists of a distribution over question tuples $x=(x_1,\ldots,x_{b_0})$, answers $a_j$ of at most $a$ bits each, and a predicate $V(x,a)$. Classical players share randomness and answer locally; quantum players share a state and answer by local measurements. $\Game$ is a quantum pseudo-telepathy game if some finite-dimensional quantum strategy wins every question and $\cval(\Game)=\gamma<1$~\cite{Liu2024}. In the GHZ game the questions are drawn uniformly from $\{000,011,101,110\}$, the answer bits must satisfy $a\oplus b\oplus c=x\lor y\lor z$, and $(\cval,\qval)=(3/4,1)$: measuring $X$ on question $0$ and $Y$ on question $1$ of $(\ket{000}+\ket{111})/\sqrt2$, and reporting eigenvalue $+1$ as bit $0$ and $-1$ as bit $1$, wins every round~\cite{Mermin1990}.}

{The classical competitor is not a set of isolated players. It is a randomized circuit with arbitrary nonuniform wiring across all $Nb$ sites, where $b$ is the total number of players, with unlimited shared randomness and unlimited computation and memory before the questions are revealed; nothing in round one is charged. Only its response to the questions is restricted: the circuit segment from the fresh question wires to the answer wires has depth at most $D$ and fan-in at most $K$, with unrestricted fan-out, and its wiring is fixed before the sites are drawn, so any site-dependent routing must be done by gates and costs depth. Depth counts sequential gate layers and fan-in counts the inputs of a gate; we call $D$ the online depth. For a fixed base game, every local strategy has a constant-size implementation. The present model also permits communication through its bounded-depth response circuit.}

\paragraph{Amplification.}
{Let $\Game^{\sqcup m}$ be the game in which $m$ copies of $\Game$ are played by $m$ disjoint sets of players and the verifier accepts only if every copy is won; it has $b=b_0m$ players. Because the players of different copies are different, once the shared randomness is fixed the winning events of the copies are functions of independent question tuples, and therefore}
\begin{equation}
 \cval(\Game^{\sqcup m})=\cval(\Game)^m=\gamma^m,
 \label{eq:main-product}
\end{equation}
{while $m$ copies of the perfect strategy win with certainty (SM Theorem~\ref{thm:sm-product}). {Equation~\eqref{eq:main-product} is the standard disjoint-player repetition used, for example, in Ref.~\cite[Definition~3.10]{CoudronStarkVidick2021}; its value factorizes exactly, so no nontrivial parallel-repetition theorem is needed.}}

{The independent objects are the verifier's question tuples for the different copies, not the processor's operations or errors. The classical device may couple all copies and share arbitrary pre-question memory. The lightcone bound counts dependencies between every pair of distinct players, including those in different copies; only when none occurs do the answers reduce to a local strategy for the disjoint-player game.}

\paragraph{Teleportation to random sites.}
{Let $H_N$ be a connected host graph with $N$ vertices and bounded degree; each player receives its own copy, with one game register at the root and two registers per edge for a Bell pair. A graph vertex is a site that can hold several registers; different edges use distinct half-edge registers. A trial proceeds as in Fig.~\ref{fig:main-protocol}(a). (i)~The verifier draws a site $T_j$ (a terminal, in the language of the SM) uniformly and independently from the $N$ vertices for every player $j$. (ii)~Round one: the device prepares the game states at the roots and Bell pairs on the edges of the selected root-to-$T_j$ paths, and Bell-measures along each path, which teleports the register to $T_j$ up to a Pauli (in general, Weyl) frame, a known local correction determined by the Bell outcomes~\cite{Bennett1993}; the measurements act on disjoint registers and are simultaneous. The device reports the outcomes and the verifier computes the frame. (iii)~Only now does the verifier reveal the questions; it delivers $x_j$, together with the frame, to site $T_j$ alone, while every other site of player $j$ receives a fixed null symbol. (iv)~Round two: the device answers at $T_j$. A quantum device performs the frame-corrected measurement; for Pauli measurements the correction is a bit flip of the recorded outcome, $r\oplus B\oplus Ax$ for frame $X^AZ^B$, where $r$ is the raw outcome bit and $x=0,1$ selects $X,Y$, respectively, so no quantum feed-forward is required. Conditioned on any sites and any transcript, the statistics at the sites equal those of the original strategy, so the ideal quantum device wins with probability one, starting from a product state and with constant depth in each round (SM Theorem~\ref{thm:sm-completeness}). The verifier computes the path frames between rounds; this classical processing is not included in device depth.}

The selected sites need not remain unknown to the device: the round-one path flags reveal them. What is fixed before their draw is the classical response wiring. An answer block can depend on at most $aK^D$ question-input blocks, so another player's independently drawn site lies in that set with probability at most $aK^D/N$. Learning the site does not provide free rewiring: any selector or routing circuit must fit within the stated depth and fan-in.

\paragraph{Main result.}
{Let $\Game$ have $b_0$ players, at most $a$ answer bits per player, classical value $\gamma<1$, and a finite-dimensional perfect quantum strategy. In the two-round test built from $\Game^{\sqcup m}$ on any connected $N$-vertex host of bounded degree, the ideal quantum device wins with probability one, and every randomized classical device whose online response has depth $D$ and fan-in $K$ wins with probability}
\begin{equation}
 p_{\rm c}\leq \gamma^m+(1-\gamma^m)
 \min\!\left\{1,\frac{a\,b\,K^D}{N}\right\},
 \quad b=b_0m.
 \label{eq:main-bound}
\end{equation}
The proof (SM Sec.~\ref{sec:sm-soundness}) counts the question blocks that can influence each selected answer. Fix a target player $j'$, its site $T_{j'}=v$, and the device's independent random tape. The backward lightcones of its at most $a$ answer wires contain at most $aK^D$ input wires in total. If $S_{j,j'}$ is the set of sites of another player $j$ whose question blocks meet these lightcones, then $\sum_{j\ne j'}|S_{j,j'}|\leq aK^D$. Each $T_j$ remains uniform, so the probability that any other player's question reaches this answer is at most $\sum_{j\ne j'}|S_{j,j'}|/N\leq aK^D/N$. Averaging over the target site and random tape and taking a union bound over the $b$ targets gives a collision probability $q\leq abK^D/N$. We do not condition on the round-one transcript in this count, because it can depend on the sites. When no collision occurs, fixing the sites and pre-question data gives a local strategy for $\Game^{\sqcup m}$, which wins with probability at most $\gamma^m$ by Eq.~\eqref{eq:main-product}; allowing success on every collision gives $p_{\rm c}\leq(1-q)\gamma^m+q$. The count includes players in different copies: the two-copy GHZ example in SM Remark~\ref{rem:sm-crosscopy} shows why those dependencies cannot be omitted.

{Choosing $m=\lceil\log_{1/\gamma}N\rceil$ gives $\gamma^m\leq1/N$ and}
\begin{equation}
 p_{\rm c}\leq\frac{1+a\,b_0\,K^D\,(\log_{1/\gamma}N+1)}{N},
 \label{eq:main-asymptotic}
\end{equation}
{which is $O(\log N/N)$ at fixed depth and tends to zero for every $D=o(\log N)$ at fixed $K$. Conversely, a classical circuit of fan-in two and depth $O(\log N)$ wins with certainty by collecting each sparse question through an OR tree over its $N$ blocks (SM Prop.~\ref{prop:sm-upper}). {Thus the classical depth scaling $\Theta(\log N)$ is tight up to constants in this amplified regime; optimality of the finite winning-probability bound is not claimed.} If the classical circuit is also required to be geometrically local on a path or grid host, its depth must be comparable to the host diameter at fixed $m$ (SM Sec.~\ref{sec:sm-geometry}).}

\paragraph{Proposed experiment.}
{We use the GHZ game because each player requires one qubit, chooses between $X$ and $Y$ measurements, and applies a bit-flip frame correction. For a tree host with $N$ vertices per player, the register (unencoded data and resource qubits, excluding readout ancillas and routing overhead) has}
\begin{equation}
 Q=3m(2N-1)
 \label{eq:main-qubits}
\end{equation}
{qubits: $3m$ game qubits and one Bell pair per edge. Table~\ref{tab:main-benchmarks} lists instances against fan-in-two classical circuits and the thresholds they set for each online depth; Fig.~\ref{fig:main-protocol}(b) shows how the threshold falls with qubit number.}

The 99-qubit instance ($m=1$, $N=17$) has thresholds $B_{\rm c}=57/68\simeq0.8382$ for depth one and $63/68\simeq0.9265$ for depth two. Both exceed the GHZ value $3/4$ because the competitor may wire all $51$ sites together. This single-copy instance implements product-state initialization, parallel Bell-pair preparation, simultaneous multi-hop Bell measurements, frame bookkeeping, question delivery after the transcript, and corrected $X/Y$ readout. For $m=2$, the base classical value decreases from $3/4$ to $9/16$ before the finite-$N$ correction; at $N=512$ the thresholds are $0.573$ for depth one and $0.583$ for depth two.

{One trial consists of: (1)~reset all registers; (2)~draw the $3m$ sites; (3)~prepare the GHZ states at the roots and Bell pairs on the edges of the selected paths---preparing every edge is allowed but unnecessary, and unused qubits remain in $\ket0$; (4)~perform the selected Bell measurements simultaneously and record their outcomes; (5)~only then reveal the questions and deliver them to the selected sites; (6)~measure $X$ for question $0$ or $Y$ for question $1$ there and apply the bit-flip correction to the record; (7)~count the trial as a win only if every copy satisfies the parity rule. A missing outcome, leakage event, timeout, or unheralded loss is a loss. No postselection is used; any postselected acceptance rule would have to be declared in advance and Eq.~\eqref{eq:main-bound} rederived for it.}

{Because Eq.~\eqref{eq:main-bound} depends on the host graph only through $N$ and the classical restriction applies only to round two, the host may be chosen for the quantum hardware. In the noise model below, attenuation is determined by the number of teleportation hops. For $N=17$, a balanced binary tree has mean path length $42/17\approx2.5$, a degree-four tree of depth two ($1+4+12$ vertices) has $28/17\approx1.6$, and a star has one hop for every non-root site. The degree-four tree retains the bounded-degree guarantee of constant quantum depth. A star loses that guarantee for its measurement-selection circuit as $N$ grows. The winning probability is $(1+V)/2$, where $V$ is the mean signed parity correlator of the state measured at the sites after preparation, teleportation, and readout, so the depth-one threshold requires $V>23/34\simeq0.6765$, the depth-two threshold requires $V>29/34\simeq0.8529$, and a target winning probability of $0.96$ requires $V\approx0.92$. For a measured preparation-and-readout visibility $V_0=0.95$, the scalar attenuation model in SM Sec.~\ref{sec:sm-99qubit} gives targets of about $0.4\%$ correlator attenuation per teleportation for the binary tree and about $1\%$ for the star. These are model-dependent estimates, not general gate-infidelity requirements.}

Using the gate and readout error rates reported for the GHZ experiment of Ref.~\cite{Kumar2025Contextuality}, the explicit model in SM Sec.~\ref{sec:sm-99-noise} gives success probabilities of $0.879$ for the 99-qubit binary tree, $0.908$ for the degree-four tree, and $0.935$ for the star, before additional idling and routing errors. All three model estimates exceed the depth-one threshold, and the star estimate also exceeds the depth-two threshold. These are model estimates, not measured scores for the protocol.

\begin{table}[t]
\caption{{GHZ instances against fan-in-two ($K=2$) classical circuits. $B_{\rm c}$ is the right-hand side of Eq.~\eqref{eq:main-bound} for online depth $D$; an entry of $1$ means that the instance does not constrain that depth. Doubling $N$ at fixed $m$ preserves the same bound while increasing $D$ by one.}}
\label{tab:main-benchmarks}
\centering
\begin{ruledtabular}
\begin{tabular}{rrrcccc}
$m$ & $N$ & $Q$ & $D=1$ & $D=2$ & $D=3$ & $D=4$\\
\hline
1 & 17 & 99 & 0.8382 & 0.9265 & 1 & 1\\
1 & 64 & 381 & 0.7734 & 0.7969 & 0.8438 & 0.9375\\
2 & 256 & 3066 & 0.5830 & 0.6035 & 0.6445 & 0.7266\\
2 & 512 & 6138 & 0.5728 & 0.5830 & 0.6035 & 0.6445\\
3 & 2048 & 36855 & 0.4270 & 0.4320 & 0.4422 & 0.4625\\
\end{tabular}
\end{ruledtabular}
\end{table}

\paragraph{Certification.}
{Let $p_{\rm exp}$ denote the true per-trial winning probability. Let $\widehat p$ be the fraction of $M$ trials, with $M$ fixed in advance, in which every copy is won. For independent, identically distributed trials, Hoeffding's inequality gives the one-sided bound}
\begin{equation}
 L_\alpha=\widehat p-
 \sqrt{\frac{\ln(1/\alpha)}{2M}},
 \qquad
 \Prb[p_{\rm exp}\geq L_\alpha]\geq1-\alpha,
 \label{eq:main-confidence}
\end{equation}
{and advantage over the declared classical class is certified when $L_\alpha>B_{\rm c}$. At $99\%$ confidence, an observed fraction at least $0.96$ on the 99-qubit instance exceeds the depth-one threshold with $156$ trials; a fraction at least $0.94$ does so with $223$ trials. The corresponding depth-two counts are $2049$ and $12580$. A margin of $0.03$ above a declared threshold requires $2559$ trials. No noise model and no independence between copies is assumed: state-preparation, Bell-measurement, idling, and readout errors all enter through $\widehat p$. Per-copy and per-path-length statistics are useful for hardware development but must not be used to discard trials. For correlated trials, a predeclared martingale bound applies provided that fresh independent sites and questions and the same circuit restriction on every trial keep its conditional classical success at most $B_{\rm c}$} (SM Sec.~\ref{sec:sm-experiment}).

\paragraph{Discussion.}
A successful experiment shows that a quantum circuit starting from a product state and having constant depth in each of two rounds produces correlations unattainable by any classical circuit with arbitrary wiring and precomputation whose response to the questions has depth $D$ and fan-in $K$. A classical circuit of depth $\Theta(\log N)$ wins with certainty, so the separation concerns the stated depth restriction. A wall-clock interpretation additionally requires a gate-time model. The test assumes a trusted verifier and a quantum implementation satisfying the stated circuit restrictions; it is not device independent. The executed circuit must therefore be reported with the data. SM Sec.~\ref{sec:sm-certification-scope} gives the full interpretation and reporting requirements.

{Relative to a GHZ or Bell test, the assumption of spacelike separation is replaced by a depth restriction on a globally wired competitor, which defines the classical comparison class for a test on a single quantum processor; relative to the benchmark experiments of Refs.~\cite{Daniel2022,Than2026}, the competitor is restricted neither to the quantum device's connectivity nor to isolated players. {Vanishing classical success was already obtained by Le Gall and by Coudron, Stark, and Vidick~\cite{LeGall2019,CoudronStarkVidick2021}. Here the quantitative improvement is the linear player prefactor in the collision bound, with explicit finite thresholds from Eq.~\eqref{eq:main-bound}:} for GHZ, the register size is $Q\approx2bN$, and the collision probability is bounded by $\min\{1,abK^D/N\}$. The verifier is trusted to draw sites and questions independently and to withhold the questions until round two; what the proof uses is that no device can access them earlier, which for a fixed public circuit is a structural property rather than a timing one.}

{The ideal theorem is noiseless and gives no fault-tolerance threshold: a constant error per hop degrades the winning probability with path length, so reaching the $m=\Theta(\log N)$ regime at constant physical noise requires {protection against accumulated errors; fault-tolerant encoded constructions such as Ref.~\cite{BGKT2020} provide one route}. The finite test requires no noise model: certification follows when a valid lower confidence bound on the winning probability exceeds the classical threshold. Open questions include a one-round version for general measurements, separations from stronger classical classes such as $\mathsf{AC}^0$ by switching-lemma methods, and small-input perfect games with classical values well below $3/4$ that could reduce the resources required by the contextuality construction in Fig.~\ref{fig:sm-contextuality}.}

\paragraph{Contextuality-based games.} 
{As an application of known KS-to-game constructions, quantum measurement contextuality~\cite{KochenSpecker1967,Budroni2022} supplies examples for the circuit test.} A finite-dimensional example of state-dependent contextuality, consisting of a state $\ket\psi$ and finitely many rank-one projectors violating a noncontextuality inequality, can be extended to a finite Kochen--Specker (KS) set containing the state and measurement projectors~\cite{Cabello2021}. If the initial projectors already contain a KS set, we can use that set directly. The KS set gives a bipartite perfect quantum strategy~\cite{Cabello2025BPQS,TrandafirCabello2025}, to which our amplification and shallow-circuit construction apply. Criticality of the KS set is not required. This construction extends the measurement set; it does not assert that the initial violation is necessary for, or quantitatively preserved by, the resulting game. Figure~\ref{fig:sm-contextuality} summarizes these steps, and SM Sec.~\ref{sec:sm-contextuality} gives an explicit qutrit example.

\paragraph{Acknowledgments.} K.B.\ was supported by a Hartree Fellowship from the Joint Center for Quantum Information and Computer Science (QuICS) at the University of Maryland, College Park. A.C.\ was supported by the AEI/MICIU Project ``Quantum Correlations for Quantum Information and Computation'' (QCQIC) (Project No.~PID2025-174467NB-I00). Generative artificial intelligence tools were used to assist with developing ideas, language editing, organization, and the preparation of portions of the manuscript.

\makeatletter
\newcommand{\KBRefChg}[2]{\expandafter\let\csname KBRefChg@#1@#2\endcsname\@empty}
\KBRefChg{Budroni2022}{title}
\KBRefChg{TrandafirCabello2025}{title}
\KBRefChg{Cabello2025KS}{title}
\KBRefChg{Lalonde2026}{eprint}
\KBRefChg{Renner2026}{eprint}
\KBRefChg{Aasnaess2022}{eprint}
\KBRefChg{Pauwels2026}{title}
\KBRefChg{Pauwels2026}{eprint}
\KBRefChg{Bennett1993}{title}
\KBRefChg{YuOh2012}{title}
\KBRefChg{BhartiJain2023}{eprint}
\KBRefChg{GrierSchaeffer2020}{title}
\KBRefChg{GrierJuSchaeffer2021}{title}
\KBRefChg{GrierJuSchaeffer2021}{eprint}
\KBRefChg{KahanamokuMeyer2022}{title}
\KBRefChg{OliveiraBarbosaGalvao2024}{author}
\KBRefChg{PirnayJerbiSeifertEisert2024}{eprint}
\KBRefChg{GrewalKumar2025}{eprint}
\KBRefChg{GriloKashefiMarkhamOliveira2026}{title}
\KBRefChg{GriloKashefiMarkhamOliveira2026}{author}
\KBRefChg{HsiehOliveiraSubramanianZhang2026}{eprint}
\KBRefChg{Than2026}{eprint}
\KBRefChg{WuZhangZhangWangYuan2026}{eprint}
\KBRefChg{Kumar2025Contextuality}{eprint}
\let\KBOriginalBibitem\bibitem
\renewcommand{\bibitem}[2][]{\def\KBCurrentBibKey{#2}\KBOriginalBibitem[#1]{#2}}
\providecommand{\bibinfo}[2]{#2}
\let\KBOriginalBibinfo\bibinfo
\renewcommand{\bibinfo}[2]{\ifcsname KBRefChg@\KBCurrentBibKey @#1\endcsname{\KBOriginalBibinfo{#1}{#2}}\else\KBOriginalBibinfo{#1}{#2}\fi}
\providecommand{\Eprint}[2]{\href{#1}{#2}}
\let\KBOriginalEprint\Eprint
\renewcommand{\Eprint}[2]{\ifcsname KBRefChg@\KBCurrentBibKey @eprint\endcsname\KBOriginalEprint{#1}{{#2}}\else\KBOriginalEprint{#1}{#2}\fi}
\makeatother
\bibliography{Ref}
\let\bibitem\KBOriginalBibitem
\let\bibinfo\KBOriginalBibinfo
\let\Eprint\KBOriginalEprint

\clearpage
\onecolumngrid
\makeatletter
\renewcommand{\theHsection}{supp.\arabic{section}}
\renewcommand{\theHequation}{supp.\arabic{equation}}
\renewcommand{\theHfigure}{supp.\arabic{figure}}
\renewcommand{\theHtable}{supp.\arabic{table}}
\providecommand{\theHtheorem}{}
\renewcommand{\theHtheorem}{supp.\arabic{theorem}}
\makeatother
\setcounter{secnumdepth}{3}
\setcounter{section}{0}
\setcounter{equation}{0}
\setcounter{figure}{0}
\setcounter{table}{0}
\setcounter{theorem}{0}
\renewcommand{\thesection}{S\arabic{section}}
\renewcommand{\thesubsection}{S\arabic{section}.\arabic{subsection}}
\makeatletter
\renewcommand{\p@subsection}{}
\makeatother
\renewcommand{\theequation}{S\arabic{equation}}
\renewcommand{\thefigure}{S\arabic{figure}}
\renewcommand{\thetable}{S\arabic{table}}
\renewcommand{\thetheorem}{S\arabic{theorem}}

\begin{center}
{\large\bfseries Supplemental Material for\\[2mm]
``Experimentally Testable Quantum Advantage in Shallow Circuits''}
\end{center}
\vspace{1mm}

This Supplemental Material provides the proofs, protocol details, and resource estimates used in the Letter. Section~\ref{sec:sm-related} discusses the scope of the claim and previous work. Sections~\ref{sec:sm-prelim}--\ref{sec:sm-soundness} define the games and circuits and prove exact amplification, quantum completeness, and classical soundness. Sections~\ref{sec:sm-amplified}--\ref{sec:sm-geometry} derive the asymptotic and geometric consequences. Section~\ref{sec:sm-games} gives the base games. Sections~\ref{sec:sm-experiment}--\ref{sec:sm-99qubit} give the experimental protocols, resource counts, host choice, error budget, and statistical tests. Section~\ref{sec:sm-limitations} states the assumptions, limitations, and open questions. Section~\ref{sec:sm-contextuality} gives the contextuality construction and example.

\section{Scope of the claim and relation to previous work}
\label{sec:sm-related}

\subsection{Finite-size separation problem}

The result addresses the following finite-size problem. Start with a finite nonlocal game $\Game$ that has a finite-dimensional perfect quantum strategy and classical value $\gamma<1$. We ask whether one can construct, for every size parameter $N$, an experimentally specified task satisfying all of the following conditions:
\begin{enumerate}
\item The quantum device starts from a computational-basis product state, uses bounded-fan-in gates, and has constant depth in each of two rounds;
\item The quantum winning probability is exactly one in the ideal model;
\item A randomized classical device may have arbitrary nonuniform wiring, shared randomness, and arbitrary information produced before the game questions are {revealed}, but its second-round response has fan-in $K$ and depth $D$;
\item The success probability of every such classical device is bounded by an explicit function of $N,D,K$, and the base-game parameters; and
\item This classical bound tends to zero for fixed $D$ when the repetition parameter is chosen as $m=\Theta(\log N)$.
\end{enumerate}
Theorem~\ref{thm:sm-amplified} answers this question affirmatively. Inequality~\eqref{eq:main-bound} gives a threshold for each finite experiment.

\subsection{Classical comparison and assumptions}
\label{sec:sm-meaning}

The term quantum advantage is used for several distinct comparisons. Random-circuit sampling and photonic sampling experiments compare a quantum processor with the strongest available classical simulation methods~\cite{Arute2019,Zhong2020,Wu2021,Madsen2022}; their complexity-theoretic foundations and verification issues are reviewed in Refs.~\cite{HarrowMontanaro2017,HangleiterEisert2023}. These experiments are major demonstrations of control over large quantum systems, but their broad asymptotic interpretation is not a proved lower bound against unrestricted classical computation. Classically verifiable protocols based on computational Bell tests form another nearby line and use cryptographic assumptions~\cite{KahanamokuMeyer2022}. More generally, Kalai, Lombardi, Vaikuntanathan, and Yang showed how to turn any multiprover nonlocal game into a single-prover interactive protocol using quantum homomorphic encryption~\cite{KalaiLombardiVaikuntanathanYang2023}. Natarajan and Zhang and subsequent work studied the quantum soundness of these compiled games~\cite{NatarajanZhang2023,KulpeMalavoltaPaddockSchmidtWalter2025}. These cryptographic constructions remove the need for spatially separated players, but their soundness is computational and rests on cryptographic assumptions.

``Unconditional'' means that, once the classical circuit model is fixed, the upper bound follows from a proved finite-game inequality, independent random choices, and circuit causality. No assumption such as $\mathsf{BPP}\neq\mathsf{BQP}$, the noncollapse of the polynomial hierarchy, or the hardness of a cryptographic primitive is used. The comparison is with classical circuits of the specified online depth and fan-in. Section~\ref{sec:sm-certification-scope} states what an experiment certifies and which implementation assumptions it requires.

\subsection{Unconditional shallow-circuit separations}

\emph{Relations, average-case tasks, and interaction.}
Bravyi, Gosset, and K\"onig introduced an explicit relation solved by constant-depth bounded-fan-in quantum circuits but not by comparable classical circuits~\cite{BGK2018}. Theorem~2 in ~\cite{BGK2018} states that, for sufficiently large $N$, a fan-in-$K$ classical circuit that succeeds with probability greater than $7/8$ on every size-$N$ instance must have depth at least $\log N/(8\log K)$. This worst-case guarantee is distinct from an average-score threshold for a specified experimental question distribution.
 Le Gall and, independently, Coudron, Stark, and Vidick developed average-case and low-depth-randomness variants~\cite{LeGall2019,CoudronStarkVidick2021}. {Their bounds already give classical success exponentially small in a positive power of the size parameter~\cite[Theorem~1]{LeGall2019},~\cite[Theorem~1.2]{CoudronStarkVidick2021}.} Bene Watts, Kothari, Schaeffer, and Tal strengthened the classical side to unbounded-fan-in $\mathsf{AC}^0$ circuits for a related relation problem~\cite{WattsKothariSchaefferTal2019}. Grier and Schaeffer showed that shallow Clifford circuits become more powerful in interactive settings, {with unconditional separations from $\mathsf{AC}^0[p]$ and conditional separations from $\mathsf{NC}^1$ and log-space computation~\cite{GrierSchaeffer2020}}, and Grier, Ju, and Schaeffer studied a noisy interactive version~{\cite{GrierJuSchaeffer2021}}.

\emph{Noise, corruption, and geometry.}
Bravyi, Gosset, K\"onig, and Tomamichel proved that an unconditional advantage can survive local stochastic noise using a three-dimensional fault-tolerant construction and also gave a simpler ideal one-dimensional result~\cite{BGKT2020}. Hasegawa and Le Gall considered arbitrary corruption of a constant fraction of the output qubits~\cite{HasegawaLeGall2021}. Bharti and Jain obtained a one-dimensional local separation with noise and sublinear classical depth when input-independent entangled advice is available~\cite{BhartiJain2023}. {Caha, Coiteux-Roy, and K\"onig later proved a noisy three-dimensional local separation from subexponential-size $\mathsf{AC}^0$ circuits, with a quantum-classical success gap arbitrarily close to one~\cite{CahaCoiteuxRoyKoenig2026}.} These results obtain stronger noise or gate-class guarantees than the present theorem, while using constructions tailored to those guarantees.

\emph{Experiments and game-based benchmarks.}
Nonlocal games have also been used directly as hardware benchmarks. Daniel {\em et al.} implemented cyclic-cluster-state games that compare shallow quantum and classical circuits with the same nearest-neighbor connectivity~\cite{Daniel2022}. Xu {\em et al.} demonstrated quantum pseudotelepathy~\cite{Xu2022}; Sheng {\em et al.} experimentally converted contextuality into nonlocality~\cite{Sheng2025}; Drmota {\em et al.} violated the classical value of the odd-cycle game with separated trapped-ion nodes~\cite{Drmota2025}; and Than {\em et al.} exceeded the classical value of a graph-coloring game on a trapped-ion processor~\cite{Than2026}. These experiments establish nonclassical correlations under a no-communication condition or compare circuits with a fixed geometry. Our proposed test asks a different question: can a device beat every globally wired classical circuit with a declared online depth and fan-in? {The amplification step can make the corresponding finite threshold substantially smaller than the classical value of one base game.}

\emph{Other resources and stronger classical classes.}
{Aasn{\ae}ss developed a general graph-distributed shallow-circuit construction from contextuality and nonlocal games, including teleportation-based quantum completeness and a lightcone/contextual-fraction classical bound~\cite{Aasnaess2022}. The random-terminal compiler used here specializes that construction.} Caha, Coiteux-Roy, and K\"onig gave an unconditional separation based on single-qubit gate teleportation rather than a Bell inequality~\cite{CahaCoiteuxRoyKoenig2024}. De Oliveira, Barbosa, and Galv{\~a}o studied advantage in temporally flat measurement-based computation~\cite{OliveiraBarbosaGalvao2024}, while Zhang, Pan, and Liu isolated quantum magic as a necessary resource for a shallow-circuit pseudo-telepathy task~\cite{ZhangPanLiu2024}. {Their construction uses a two-round teleportation protocol to separate general quantum circuits from Clifford circuits.} De Oliveira, Subramanian, Mendes, and Hsieh established noisy qudit separations from biased threshold circuits~\cite{OliveiraSubramanianMendesHsieh2025}. Grewal and Kumar and, in a related direction, Grilo, Kashefi, Markham, and de Oliveira obtained separations from more expressive constant-depth classical circuit classes~\cite{GrewalKumar2025,GriloKashefiMarkhamOliveira2026}. Hsieh, de Oliveira, Subramanian, and Zhang proved a depth hierarchy within shallow quantum circuits while also obtaining classical separations~\cite{HsiehOliveiraSubramanianZhang2026}.

\emph{Sampling, random circuits, and learning.}
Bene Watts and Parham proved an unconditional shallow-circuit sampling separation under a restriction on the classical randomness, and Grier, Kane, Morris, Ostuni, and Wu later obtained a depth-seven geometrically local sampling separation against bounded-fan-in classical circuits with arbitrary product-distributed inputs~\cite{WattsParham2026,GrierKaneMorrisOstuniWu2026}. Random shallow circuits can also give unconditional advantage through measurement-induced entanglement~\cite{WattsGossetLiuSoleimanifar2025,McGinleyHoMalz2025}. Related results establish unconditional separations for shallow-circuit learning and distribution learning~\cite{ZhangGongLiDeng2024,PirnayJerbiSeifertEisert2024}. Mid-circuit measurements and feed-forward also appear in recent measurement-driven and adaptive shallow-circuit proposals~\cite{CaoEisert2026,WuZhangZhangWangYuan2026}. Those proposals use standard sampling or discrete-logarithm hardness assumptions; by contrast, the bound proved here is unconditional for the stated depth- and fan-in-restricted classical model.

These results differ in the classical circuit classes they exclude, their noise tolerance, their use of inputs and interaction, and the complexity of their quantum circuits. {Our shared ancestor count improves the quadratic player prefactor in Ref.~\cite[Eq.~(4.85)]{Aasnaess2022} to linear dependence. For a fixed base game, that bound with standard disjoint-player repetition already gives $p_{\rm c}\leq\gamma^m+O(m^2K^D/N)$. At $m=\lceil\log_{1/\gamma}N\rceil$, this improves $O(K^D\log^2N/N)$ to $O(K^D\log N/N)$; both already vanish. For one-bit GHZ, the published coarse estimate $\min\{1,b^2K^D/N\}$ is one at $b=3$, $N=17$, $K=2$, and $D=1,2$; the linear count gives nontrivial 99-qubit thresholds. This comparison concerns the published estimates, not optimal classical success.}

\section{Finite games, local fraction, and circuit models}
\label{sec:sm-prelim}
\subsection{Nonlocal games}
\begin{definition}[Finite nonlocal game]
A finite $b$-player game is a tuple
\begin{equation}
 \Game=(\mathcal X_1,\ldots,\mathcal X_b,
 \mathcal A_1,\ldots,\mathcal A_b,\mu,V),
\end{equation}
where $\mu$ is a distribution over question tuples $x=(x_1,\ldots,x_b)$ and
\begin{equation}
 V(x,a)\in\{0,1\}
\end{equation}
is the acceptance predicate for answers $a=(a_1,\ldots,a_b)$.
\end{definition}

A classical strategy consists of shared randomness $\lambda$, independent of $x$, and local response distributions $p_j(a_j|x_j,\lambda)$. Its value is
\begin{equation}
 \cval(\Game)=\max
 \sum_{x,a,\lambda}\mu(x)p(\lambda)
 \prod_{j=1}^{b}p_j(a_j|x_j,\lambda)V(x,a).
\end{equation}
Because this expression is affine in the shared-randomness distribution, deterministic local response functions attain the optimum.

A finite-dimensional quantum strategy consists of a state $\rho$ on $\bigotimes_j\mathcal H_j$ and local {positive-operator-valued measures (POVMs)} $\{M^{(j)}_{a_j|x_j}\}_{a_j}$. {For each question, these measurement operators are positive and sum to the identity; the following trace gives the outcome probability.} Its conditional behavior is
\begin{equation}
 p(a|x)=\Tr\!\left[
 \rho\bigotimes_{j=1}^{b}M^{(j)}_{a_j|x_j}
 \right].
\end{equation}
We call $\Game$ a finite-dimensional pseudo-telepathy game when a finite-dimensional strategy attains $\qval(\Game)=1$ and $\cval(\Game)=\gamma<1$.

\subsection{Local and nonlocal fraction}

{The decomposition of a behavior into local and residual parts was introduced in the study of the local content of quantum correlations by Elitzur, Popescu, and Rohrlich~\cite{Elitzur1992}.} We use the following form, which also applies to signaling behaviors produced by a globally wired classical circuit.

\begin{definition}[Nonlocal fraction]
For a conditional behavior $p(a|x)$, define $\NLF(p)$ as the smallest $\alpha\in[0,1]$ for which
\begin{equation}
 p=(1-\alpha)p_{\rm L}+\alpha p_{\rm R},
 \label{eq:sm-nlf}
\end{equation}
where $p_{\rm L}$ is local and $p_{\rm R}$ is an arbitrary conditional behavior.
\end{definition}

For nonsignaling  behaviors this is the contextual fraction of Ref.~\cite{Abramsky2017}. The proof below only needs the elementary convex decomposition.

\begin{lemma}[Resource inequality]
\label{lem:sm-resource}
For every game $\Game$ and behavior $p$,
\begin{equation}
 \Prb_p[\mathrm{win}\ \Game]
 \leq \cval(\Game)+
 \bigl[1-\cval(\Game)\bigr]\NLF(p).
 \label{eq:sm-resource}
\end{equation}
\end{lemma}

\begin{proof}
Take a decomposition in Eq.~\eqref{eq:sm-nlf} with $\alpha=\NLF(p)$. Game success is linear in $p$, the local part wins with probability at most $\cval(\Game)$, and the residual wins with probability at most one. Therefore
\begin{align}
 \Prb_p[\mathrm{win}]
 &\leq(1-\alpha)\cval(\Game)+\alpha\\
 &=\cval(\Game)+[1-\cval(\Game)]\alpha.
\end{align}
\end{proof}

\subsection{Classical circuits and lightcones}

A Boolean circuit is a directed acyclic graph whose internal gates have fan-in at most $K$. Its depth is the largest number of gates on an input-to-output path. Fan-out is unrestricted. For an output wire $z$, let $\LC(z)$ be the set of original input wires having a directed path to $z$.

\begin{lemma}[Fan-in lightcone bound]
\label{lem:sm-lightcone}
If a circuit has depth $D$ and fan-in $K$, then
\begin{equation}
 |\LC(z)|\leq K^D.
\end{equation}
For an answer block containing at most $a$ output wires, the union of its backward lightcones has at most $aK^D$ input wires.
\end{lemma}

\begin{proof}
At depth zero, a wire has at most one original ancestor. If every wire after $d$ layers has at most $K^d$ ancestors, a gate in the next layer takes at most $K$ inputs and therefore has at most $K^{d+1}$ original ancestors. Induction gives the first claim; a union bound over the $a$ answer wires gives the second.
\end{proof}

For a host graph $H$, a range-$r$ circuit places every wire at a graph vertex and allows a gate only when all input and output locations are within graph distance $r$ of one another. Let
\begin{equation}
 B_H(v,R)=\{u:\dist_H(u,v)\leq R\},\qquad
 \beta_H(R)=\max_v|B_H(v,R)|.
\end{equation}

\begin{lemma}[Geometric lightcone bound]
\label{lem:sm-geo-lightcone}
In a range-$r$, depth-$D$ circuit, every input in the backward lightcone of an output at $v$ is located in $B_H(v,rD)$.
\end{lemma}

\begin{proof}
After one layer, information moves by at most distance $r$. Applying the triangle inequality inductively over $D$ layers gives distance at most $rD$.
\end{proof}

\section{Exact amplification with disjoint players}
\label{sec:sm-product}

{{We recall the standard disjoint-player product construction used in Ref.~\cite[Definition~3.10]{CoudronStarkVidick2021}. Each copy uses a disjoint set of players.} Once shared randomness is fixed, the copy-winning events are functions of independent question tuples, so the value factorizes directly. No nontrivial parallel-repetition theorem is invoked.}

\begin{definition}[Disjoint-player repetition]
Let $\Game$ have player set $[b_0]$. The game $\Game^{\sqcup m}$ has player set $[m]\times[b_0]$. For every copy $s\in[m]$, the referee independently samples $x^{(s)}\sim\mu$, sends $x_j^{(s)}$ to the distinct player $(s,j)$, and receives $a_j^{(s)}$. The verifier accepts iff every copy satisfies the base predicate.
\end{definition}

\begin{theorem}[Exact disjoint-player product]
\label{thm:sm-product}
For every finite game and integer $m\geq1$,
\begin{equation}
 \cval(\Game^{\sqcup m})=\cval(\Game)^m.
 \label{eq:sm-product}
\end{equation}
If $\Game$ has a finite-dimensional perfect quantum strategy, then so does $\Game^{\sqcup m}$.
\end{theorem}

\begin{proof}
Write $\gamma=\cval(\Game)$. First fix a deterministic strategy for $\Game^{\sqcup m}$. Player $(s,j)$ appears only in copy $s$, and therefore its answer is a function only of $x_j^{(s)}$. Let $W_s$ be the event that copy $s$ is won. The event $W_s$ depends only on the independently sampled question tuple $x^{(s)}$, hence the events $W_1,\ldots,W_m$ are independent. The responses in copy $s$ form a valid deterministic strategy for the base game, so $\Prb(W_s)\leq\gamma$. Thus
\begin{equation}
 \Prb(W_1\cap\cdots\cap W_m)
 =\prod_{s=1}^{m}\Prb(W_s)\leq\gamma^m.
\end{equation}
Conditioning on shared randomness makes a randomized strategy deterministic, so averaging preserves the upper bound. Conversely, using an optimal base strategy independently in every copy attains $\gamma^m$.

For the quantum statement, give every copy a separate tensor factor of the base state and let its distinct players use the base measurements. Every copy wins with certainty.
\end{proof}

In the compiled processor, this product is invoked only after conditioning on the sites, random tape, and pre-question data for which no cross-player lightcone collision occurs. It does not assume that the device runs independent copies; shared memory and cross-copy wiring are allowed by the model in Sec.~\ref{sec:sm-soundness}.

For $m$ GHZ copies, $b=3m$ and the values are $(\cval,\qval)=((3/4)^m,1)$. For $m$ Magic-Square copies, $b=2m$ and the values are $((8/9)^m,1)$. For the qutrit $5\times9$ game, the values are $((44/45)^m,1)$.

\section{Graph teleportation and exact quantum completeness}
\label{sec:sm-teleport}

\subsection{Teleportation convention}

We fix the phases in the standard teleportation protocol~\cite{Bennett1993} explicitly because the measurement correction in round two depends on them. Fix a local dimension $d\geq2$, let $\omega=e^{2\pi i/d}$, and define
\begin{equation}
 X\ket j=\ket{j+1\bmod d},\qquad
 Z\ket j=\omega^j\ket j.
\end{equation}
Let
\begin{equation}
 \ket\Phi=\frac{1}{\sqrt d}\sum_{j=0}^{d-1}\ket j\ket j,
 \qquad
 \ket{\beta_{a,b}}=(\id\otimes X^aZ^b)\ket\Phi.
\end{equation}
The $d^2$ Bell states form an orthonormal basis. Indeed,
\begin{equation}
 \langle\beta_{a,b}|\beta_{a',b'}\rangle
 =\delta_{a,a'}\frac{1}{d}
 \sum_{j=0}^{d-1}\omega^{(b'-b)j}
 =\delta_{a,a'}\delta_{b,b'}.
\end{equation}

\begin{lemma}[Teleportation identity]
\label{lem:sm-teleport}
For every $\ket\psi=\sum_j\psi_j\ket j$,
\begin{equation}
 (\langle\beta_{a,b}|_{AB}\otimes\id_C)
 (\ket\psi_A\otimes\ket\Phi_{BC})
 =\frac{1}{d}X^aZ^{-b}\ket\psi_C.
 \label{eq:sm-teleport}
\end{equation}
Every outcome has probability $1/d^2$.
\end{lemma}

\begin{proof}
Expand
\begin{equation}
 \ket\psi_A\ket\Phi_{BC}
 =\frac{1}{\sqrt d}\sum_{j,k}\psi_j\ket j_A\ket k_B\ket k_C
\end{equation}
and
\begin{equation}
 \langle\beta_{a,b}|_{AB}
 =\frac{1}{\sqrt d}\sum_{\ell=0}^{d-1}
 \omega^{-b\ell}\langle\ell|_A\langle\ell+a|_B.
\end{equation}
The inner product enforces $j=\ell$ and $k=\ell+a$, yielding
\begin{equation}
 \frac{1}{d}\sum_{\ell}\psi_\ell\omega^{-b\ell}\ket{\ell+a}
 =\frac{1}{d}X^aZ^{-b}\ket\psi.
\end{equation}
The squared norm is $1/d^2$.
\end{proof}

The same identity holds when the teleported register is entangled with an arbitrary reference: conditioned on outcome $(a,b)$, the endpoint state is conjugated by
\begin{equation}
 P(a,b)=X^aZ^{-b}.
\end{equation}
This follows first for a pure state $\sum_j\ket{\xi_j}_R\ket j_A$ by the same component calculation and then for a mixed state by purification.

{For a root-to-terminal path, orient every spanning-tree edge from parent to child and order each resource pair as (parent-side half, child-side half). At each hop, the Bell-basis convention above takes the incoming logical register as subsystem $A$ and the parent-side half of the next outward edge as subsystem $B$.}

\begin{lemma}[Path teleportation]
\label{lem:sm-path}
Let $v_0,v_1,\ldots,v_\ell$ be a path. Put the input qudit at $v_0$ and one maximally entangled pair on every edge. {At every hop, Bell-measure the incoming logical register with the parent-side half of the next edge, in that order;} at an internal vertex these are the child-side half of the preceding edge and the parent-side half of the following edge. If the outcomes are $s_1,\ldots,s_\ell$, the endpoint state is conjugated by
\begin{equation}
 Q=P(s_\ell)P(s_{\ell-1})\cdots P(s_1).
 \label{eq:sm-path-frame}
\end{equation}
All Bell measurements can be carried out simultaneously.
\end{lemma}

\begin{proof}
The Bell measurements act on pairwise disjoint registers and therefore commute. Their joint post-measurement state may be analyzed in path order. Applying Lemma~\ref{lem:sm-teleport} inductively composes the correction operators from left to right, giving Eq.~\eqref{eq:sm-path-frame}.
\end{proof}

Because $Z^uX^v=\omega^{uv}X^vZ^u$, every finite product of the $P(a,b)$ is, up to global phase, one of the $d^2$ operators $X^AZ^{-B}$. Thus the verifier sends a constant-size frame label when $d$ is fixed.

\begin{lemma}[Corrected measurement]
\label{lem:sm-corrected-measurement}
If the endpoint state is $P\rho P^\dagger$ and the desired POVM is $\{M_a\}_a$, measuring $\{PM_aP^\dagger\}_a$ reproduces the desired distribution.
\end{lemma}

\begin{proof}
For every outcome,
\begin{equation}
 \Tr[(P\rho P^\dagger)(PM_aP^\dagger)]
 =\Tr(\rho M_a).
\end{equation}
\end{proof}

\subsection{Hardware layout}

Let $H=(V,E,r_0)$ be a connected rooted graph with $N=|V|$ and maximum degree $\Delta$. Fix a {rooted shortest-path spanning tree}. For every terminal $v$, let $\pi(v)$ be the unique root-to-$v$ path in that tree.

For every player $j$ and edge $e=\{u,v\}$, introduce two $d$-dimensional half-edge registers, one at $u$ and one at $v$. The game register of player $j$ begins at the root. Each edge half-pair {that is used in a trial} is prepared in $\ket\Phi$. Distinct edges use distinct registers, so all edge entanglement can be prepared in parallel even when edges share a graph vertex. The local register count is $O(b\Delta)$.

\subsection{Two-round task}

\begin{construction}[Graph-distributed game]
\label{con:sm-compiler}
The verifier and device proceed as follows.
\begin{enumerate}
\item {Using private verifier randomness independent of the device's complete random tape and initial state,} independently for every player $j\in[b]$, sample a terminal $T_j$ uniformly from $V$.
\item In round one, instruct the device to Bell-measure along $\pi(T_j)$ and return all Bell labels. {An invalid, missing, or late required response causes immediate rejection and is never postselected.}
\item Compute the endpoint frame $P_j$ from Eq.~\eqref{eq:sm-path-frame}; set $P_j=\id$ when $T_j=r_0$.
\item {Using fresh private verifier randomness independent of the terminals, the device's complete random tape, the complete round-one transcript, and all pre-round-two memory,} sample the base-game question tuple $x$ {(or reveal it, if it was generated earlier and kept private; see Sec.~\ref{sec:sm-soundness})}.
\item {In round two, for each fixed block $(j,u)$, supply an input bundle
$Y_{j,u}=\operatorname{enc}(x_j,P_j)$ when $u=T_j$ and a fixed null encoding otherwise. The null encoding is distinct from every valid question-frame encoding. These bundles are the only device inputs whose values depend on $x$.} Read the answer $a_j$ at the selected terminal.
\item Accept iff $V(x,a)=1$.
\end{enumerate}
\end{construction}

The external verifier's classical processing is not charged to device depth. The result is therefore a two-round interactive shallow-device separation, not an end-to-end constant-depth algorithm for the verifier.

\begin{theorem}[Exact quantum completeness]
\label{thm:sm-completeness}
Suppose $\Game$ has a finite-dimensional strategy of success $p_{\rm q}$. Assume the gate model contains the preparation and measurement operations of this fixed strategy exactly (Sec.~\ref{sec:sm-limitations}). Then $\Comp_H(\Game)$ has a bounded-fan-in quantum implementation that starts from a product state, has $O_{\Game,\Delta}(1)$ depth in each round including resource preparation in round one, and succeeds with probability exactly $p_{\rm q}$.
\end{theorem}

\begin{proof}
Prepare a purification of the fixed game state at the root by a constant-size unitary. Prepare every edge pair in $\ket\Phi$ using constant-arity gates. Since different copies and different half-edge pairs occupy disjoint registers, all preparations run in a number of layers independent of $N$.

In round one, the path flags given by the verifier select the Bell measurements. At an internal vertex, at most a constant number of child choices must be multiplexed because $\Delta$ is bounded. Operations at different vertices and for different players use distinct registers. Hence all selected Bell measurements have constant depth. Lemma~\ref{lem:sm-path} gives the conditional joint terminal state
\begin{equation}
 (P_1\otimes\cdots\otimes P_b)\rho
 (P_1\otimes\cdots\otimes P_b)^\dagger.
\end{equation}

If the desired measurement for player $j$ on question $x_j$ is $\{M^{(j)}_{a|x_j}\}_a$, round two measures
\begin{equation}
 \{P_jM^{(j)}_{a|x_j}P_j^\dagger\}_a.
\end{equation}
For a fixed game and dimension, the control alphabet is finite and compiles into constant depth. Applying Lemma~\ref{lem:sm-corrected-measurement} to every tensor factor shows that, conditioned on every terminal tuple and transcript, the answer distribution equals that of the original game strategy. Averaging changes nothing, so the success is exactly $p_{\rm q}$.
\end{proof}

\begin{corollary}[Depth and width under repetition]
\label{cor:sm-parallel}
For fixed $\Game$, the strategy for $\Game^{\sqcup m}$ prepares $m$ independent copies and executes all corresponding operations in parallel. Quantum depth and gate fan-in are independent of $m$; the width grows by a factor $m$.
\end{corollary}

\begin{remark}[Preparation on selected paths only]
\label{rem:sm-selected}
{Since the path flags are supplied before round one, the device may prepare Bell pairs only on the edges of the selected paths, by classically controlled constant-depth preparations on the corresponding registers, and leave every other half-edge register in $\ket0$. Theorem~\ref{thm:sm-completeness} is unchanged, because the selected Bell measurements act only on prepared pairs. This reduces the number of entangled pairs per trial from $\Theta(mN)$ to $\sum_j\ell(T_j)$, where $\ell(T_j)$ is the length of $\pi(T_j)$, which is $O(m\log N)$ on a balanced tree, and avoids preparing unused Bell pairs, although idle errors and cross-talk may remain} (Sec.~\ref{sec:sm-limitations}). All $Q$ registers must nevertheless exist, so that every terminal is a distinct physical location.
\end{remark}

\section{Classical soundness from random terminals and lightcones}
\label{sec:sm-soundness}
\subsection{Classical circuit model}

{The classical device may be nonuniform and randomized. Let $R$ denote its complete random tape, including any round-two coins; $R$ is independent of the verifier's private randomness. Before the questions are revealed, the device may hold arbitrary shared randomness and arbitrary memory produced in round one. We charge the depth $D$ of the circuit segment from the fixed round-two input bundles to the answer outputs. For each input length and fixed $R=r$, the structural directed acyclic graph of this segment is fixed before the terminals and first-round transcript are revealed. Every backward lightcone below means reachability in this unspecialized structural graph, before any memory, transcript, or control-wire values are substituted. A switch, multiplexer, or universal-circuit selector that changes information flow must itself be implemented by gates and is charged to the depth and fan-in. Thus terminal-dependent rewiring is not free. The fresh question tuple $X$ is independent of the sigma-algebra generated by $(T,R)$, the round-one transcript, and the pre-round-two memory.}

{This independence is the only property of the questions that the proof uses. The verifier may therefore generate $X$ at any time, before or after round one, provided that nothing in the device's round-one operation or pre-round-two memory depends on it; ``fresh'' means inaccessible to the device before round two, not generated after round one in physical time. For a fixed public circuit, the absence of earlier access is a structural property of the circuit and needs no real-time random-number generation.}

{The fixed-wiring requirement is essential. If terminal-dependent rewiring were free, a classical device could use a zero frame transcript, route the selected GHZ question bits $x_A,x_B$ to player $C$'s selected output, and return $a=b=0$, $c=x_A\lor x_B$. On the promised tuples $000,011,101,110$, this wins with certainty using a fan-in-two, depth-one response. In the present model, selecting those sparse inputs is charged to the circuit.}

An answer contains at most $a$ output bits. There may be additional unused outputs, but the verifier reads only the answer block at the selected terminal.

\begin{definition}[Cross-player collision]
{Fix a terminal tuple $t=(t_1,\ldots,t_b)$ and complete random tape $r$. An ordered pair $(j,j')$, $j\neq j'$, is bad for $(t,r)$ if any wire in the selected bundle $Y_{j,t_j}$ lies in the unspecialized backward lightcone of an answer output at block $(j',t_{j'})$. The pair $(t,r)$ is good if no ordered player pair is bad.}
\end{definition}

\subsection{Collision probability}

\begin{lemma}[Collision bound for bounded-fan-in circuits]
\label{lem:sm-collision}
For independent uniform terminals in an $N$-vertex host graph,
\begin{equation}
 {\Prb_{T,R}[\mathrm{bad}]}
 \leq \delta_{\rm fan}:=
 \min\!\left\{1,
 \frac{abK^D}{N}
 \right\}.
 \label{eq:sm-collision}
\end{equation}
\end{lemma}

\begin{proof}
Fix a target player $j'$ and condition on $T_{j'}=v$ and $R=r$, but not on the first-round transcript, which may depend on the other terminals. Every $T_j$ with $j\neq j'$ remains uniform under this conditioning. For fixed $r$, the unspecialized response graph is independent of all terminals and wire values. For each $j\neq j'$, let $S_{j,j'}(v,r)$ be the set of vertices $u$ such that some wire in $Y_{j,u}$ is an ancestor of an answer output at $(j',v)$. These input blocks are disjoint across players and vertices. By Lemma~\ref{lem:sm-lightcone}, the selected answer block has at most $aK^D$ fresh-input ancestors in total, so
\begin{equation}
 \sum_{j\neq j'}|S_{j,j'}(v,r)|\leq aK^D.
\end{equation}
A collision into this answer occurs only if $T_j\in S_{j,j'}(v,r)$ for some $j\neq j'$. Hence
\begin{align}
 \Prb[\exists j\neq j':(j,j')\text{ bad}\mid T_{j'}=v,R=r]
 &\leq\sum_{j\neq j'}\frac{|S_{j,j'}(v,r)|}{N}\nonumber\\
 &\leq\frac{aK^D}{N}.
\end{align}
Average over $v$ and $R$, sum over the $b$ target players, and cap the result at one. All source players, including those in other copies of the base game, are included in the same ancestor count.
\end{proof}

\begin{lemma}[Geometric collision bound]
\label{lem:sm-geo-collision}
If the round-two circuit is range-$r$ local on $H$, then
\begin{equation}
 {\Prb_{T,R}[\mathrm{bad}]}
 \leq \delta_{\rm geo}:=
 \min\!\left\{1,
 \frac{b(b-1)\beta_H(rD)}{N}
 \right\}.
 \label{eq:sm-geo-collision}
\end{equation}
\end{lemma}

\begin{proof}
{Fix $j\neq j'$ and condition on $T_{j'}=v$ and $R=\rho$. Lemma~\ref{lem:sm-geo-lightcone} implies that a wire in the selected bundle $Y_{j,T_j}$ can affect the answer at $(j',v)$ only if $T_j\in B_H(v,rD)$. Independence and uniformity give probability at most $\beta_H(rD)/N$. Average over $v$ and $R$, sum over ordered pairs, and cap at one.}
\end{proof}

\subsection{Locality without cross-player collisions}

\begin{lemma}[Locality without collisions]
\label{lem:sm-good-local}
{Condition on a good pair $(T,R)=(t,r)$, the complete first-round transcript, and the pre-round-two memory. The resulting round-two answer behavior is a deterministic local strategy for the base game.}
\end{lemma}

\begin{proof}
{After conditioning, the complete device tape and all information that exists before the game questions are fixed. The question tuple is independent of these variables. Fix player $j'$. Its answer block is a deterministic function of the round-two inputs in its unspecialized backward lightcone and the fixed conditioned data. Goodness excludes every fresh $x_j$-dependent wire for $j\neq j'$. The lightcone may contain $j'$'s own selected bundle, null encodings, correction labels, and fixed memory. Therefore there exists a function $f_{j'}$ such that}
\begin{equation}
 a_{j'}=f_{j'}(x_{j'}).
\end{equation}
Repeating this for every player gives a deterministic local response tuple. The functions may depend parametrically on the conditioned data, but that data is shared randomness independent of the questions.
\end{proof}

\begin{theorem}[Classical winning-probability bound]
\label{thm:sm-master}
For a $b$-player base game of classical value $\gamma$, every randomized classical device for $\Comp_H(\Game)$ satisfies
\begin{equation}
 p_{\rm c}\leq\gamma+(1-\gamma)\delta,
 \label{eq:sm-master}
\end{equation}
where $\delta$ may be chosen as $\delta_{\rm fan}$ from Eq.~\eqref{eq:sm-collision}, or as $\delta_{\rm geo}$ from Eq.~\eqref{eq:sm-geo-collision} for a geometrically local circuit.
\end{theorem}

\begin{proof}
{Let $q=\Pr_{T,R}[\mathrm{bad}]$. On a good pair $(T,R)$, further conditioning on the transcript and memory and applying Lemma~\ref{lem:sm-good-local} bounds success by $\gamma$; averaging over those pre-question variables preserves the bound. On a bad pair, success is at most one. Hence}
\begin{equation}
 p_{\rm c}\leq(1-q)\gamma+q
 =\gamma+(1-\gamma)q.
\end{equation}
Substitute either collision bound. Equivalently, the behavior has local fraction at least $1-q$, and Lemma~\ref{lem:sm-resource} gives the same inequality.
\end{proof}

\begin{remark}[Cross-copy dependencies]
\label{rem:sm-crosscopy}
{When Theorem~\ref{thm:sm-master} is applied to $\Game^{\sqcup m}$, Lemma~\ref{lem:sm-collision} counts input dependencies from all other players into each answer block, including players in different copies. These cross-copy dependencies cannot be omitted. Consider $\Game^{\sqcup2}$ for the GHZ game and a classical device whose answers in copy $1$ may depend on the complete question tuple $x^{(2)}$ of copy $2$ and vice versa (the Boolean responses below have fan-in-two depth two when the relevant question bits are supplied directly; routing sparse inputs is additional). For every $y\in\{000,011,101,110\}$ there is a deterministic GHZ strategy $\sigma_y$ that loses exactly on $y$: with all six response bits zero except $a_0=b_0=c_0=1$ it loses only on $000$, and with all bits zero except $a_1=1$, $b_1=1$, or $c_1=1$ it loses only on $011$, $101$, or $110$, respectively. Let copy $1$ play $\sigma_{x^{(2)}}$ and copy $2$ play $\sigma_{x^{(1)}}$. Copy $1$ is won iff $x^{(1)}\neq x^{(2)}$ and copy $2$ is won iff $x^{(2)}\neq x^{(1)}$, the same event, so both copies are won with probability $3/4>9/16=\cval(\Game^{\sqcup2})$. A bound that counted only collisions between players of the same copy would wrongly assert $p_{\rm c}\leq9/16$ for this device. The linear factor $b=b_0m$ in Eq.~\eqref{eq:sm-collision} retains all these dependencies: they share the same $aK^D$ ancestor budget for each target answer.}
\end{remark}

\section{Amplification and asymptotic bounds}
\label{sec:sm-amplified}

\begin{theorem}[Amplified classical bound and quantum completeness]
\label{thm:sm-amplified}
Let $\Game$ be a fixed $b_0$-player pseudo-telepathy game with classical value $\gamma<1$, at most $a$ answer bits per player, and a finite-dimensional perfect quantum strategy. Let $(H_N)$ be a family of connected $N$-vertex host graphs with uniformly bounded degree. Then $\Comp_{H_N}(\Game^{\sqcup m})$ has a quantum strategy that starts from a product state, has constant depth in each round, and wins with probability one, while every depth-$D$, fan-in-$K$ randomized classical strategy satisfies
\begin{align}
 p_{\rm c}
 &\leq\gamma^m+(1-\gamma^m)
 \min\!\left\{1,
 \frac{ab_0mK^D}{N}
 \right\}
 \label{eq:sm-amplified-exact}\\
 &\leq\gamma^m+\frac{ab_0mK^D}{N}.
 \label{eq:sm-amplified-simple}
\end{align}
\end{theorem}

\begin{proof}
Theorem~\ref{thm:sm-product} gives repeated values $(\gamma^m,1)$. Theorem~\ref{thm:sm-completeness} and Corollary~\ref{cor:sm-parallel} give the constant-depth perfect quantum implementation. Apply Theorem~\ref{thm:sm-master} with $b=b_0m$ and base classical value $\gamma^m$. For the second line, use $1-\gamma^m\leq1$ and remove the cap on the collision term.
\end{proof}

Set
\begin{equation}
 m_N=\left\lceil\log_{1/\gamma}N\right\rceil.
\end{equation}
Then
\begin{equation}
 \gamma^{m_N}\leq N^{-1},\qquad
 m_N\leq\log_{1/\gamma}N+1.
\end{equation}
Substituting into Eq.~\eqref{eq:sm-amplified-simple} proves
\begin{equation}
 p_{\rm c}\leq\frac{1}{N}
 +\frac{ab_0K^D}{N}
 \bigl(\log_{1/\gamma}N+1\bigr).
 \label{eq:sm-near-max}
\end{equation}
For fixed $K$ and $D=o(\log N)$, $K^D=N^{o(1)}$, so $p_{\rm c}\leq N^{-1+o(1)}$. In particular, fixed depth gives $p_{\rm c}=O(\log N/N)$.

The number of player--vertex blocks is
\begin{equation}
 n=b_0m_NN=\Theta_{\Game}(N\log N).
 \label{eq:sm-size}
\end{equation}
Since $N=\Theta(n/\log n)$ and $\log N=\Theta(\log n)$, Eq.~\eqref{eq:sm-near-max} becomes
\begin{equation}
 p_{\rm c}=O_{\Game}\!\left(
 \frac{K^D\log^2n}{n}
 \right).
 \label{eq:sm-size-bound}
\end{equation}

\begin{theorem}[Classical depth lower bound]
\label{thm:sm-depth}
Fix $K\geq2$ and target success $s\in(0,1)$. For all sufficiently large $N$, a classical circuit with success at least $s$ obeys
\begin{align}
 D&\geq\log_KN-\log_K\log N-O_{\Game,K,s}(1),
 \label{eq:sm-depth-N}\\
 D&\geq\log_Kn-2\log_K\log n-O_{\Game,K,s}(1).
 \label{eq:sm-depth-n}
\end{align}
\end{theorem}

\begin{proof}
For sufficiently large $N$, $1/N\leq s/2$. Equation~\eqref{eq:sm-near-max} and $p_{\rm c}\geq s$ imply
\begin{equation}
 K^D\geq
 \frac{sN}{2ab_0(\log_{1/\gamma}N+1)}.
\end{equation}
Taking base-$K$ logarithms gives Eq.~\eqref{eq:sm-depth-N}. Equation~\eqref{eq:sm-size} implies
\begin{equation}
 \log_KN=\log_Kn-\log_K\log n-O_{\Game}(1),
\end{equation}
and $\log_K\log N=\log_K\log n+O(1)$, which gives Eq.~\eqref{eq:sm-depth-n}.
\end{proof}

\begin{proposition}[Classical strategy of logarithmic depth]
\label{prop:sm-upper}
If every supported base-game question tuple has at least one accepting answer tuple, the distributed task has a deterministic fan-in-two classical solution of depth $O(\log n)$ and polynomial size.
\end{proposition}

\begin{proof}
The device may report zero Bell labels in round one, so every correction is the identity. In round two, a constant-depth local decoder flags each possible question symbol at each terminal block. For each player and symbol, a balanced OR tree aggregates the flags over the $N$ blocks. This recovers all questions in depth $O(\log N)$. In parallel for each copy, a fixed lookup table chooses an accepting answer tuple for its recovered base-game questions. Unrestricted fan-out broadcasts each answer to every block assigned to that player, and the local terminal marker selects the unique emitted answer. The total depth is $O(\log N)=O(\log n)$.
\end{proof}

The repetition count $m=\Theta(\log N)$ is also the correct scale for this lightcone union bound. Writing $\alpha=\ln(1/\gamma)>0$, the simplified upper bound is
\begin{equation}
 f_N(m)=e^{-\alpha m}+c\frac{mK^D}{N},
 \qquad c=ab_0.
\end{equation}
For real $m$, this function is strictly convex. When $N/K^D$ is sufficiently large, its minimum occurs at
\begin{equation}
 m_* = \frac{1}{\alpha}\ln\!\left(\frac{\alpha N}{cK^D}\right).
\end{equation}
The minimizing positive integer is one of the two integers adjacent to $m_*$. Thus $m=\Theta(\log(N/K^D))$ minimizes this bound as $N/K^D\to\infty$, and $m=\Theta(\log N)$ for fixed or sublogarithmic depth.

\section{Geometrically local classical circuits}
\label{sec:sm-geometry}

Combining Theorem~\ref{thm:sm-product}, Theorem~\ref{thm:sm-completeness}, and the geometric form of Theorem~\ref{thm:sm-master} gives the following host-volume statement.

\begin{theorem}[Host-volume bound]
\label{thm:sm-host-volume}
If the classical round-two circuit is range-$r$ local on $H_N$, then
\begin{equation}
 p_{\rm c}\leq\gamma^m+(1-\gamma^m)
 \min\!\left\{1,
 \frac{(b_0m)(b_0m-1)\beta_{H_N}(rD)}{N}
 \right\}.
 \label{eq:sm-host-volume}
\end{equation}
On bounded-degree host families, the quantum strategy uses graph-local gates, starts from a product state, has constant depth per round, and wins with probability one.
\end{theorem}

For the path $P_N$, every radius-$R$ ball has at most $2R+1$ vertices. If a range-$r$ circuit succeeds with probability at least $\gamma^m+\varepsilon$, Eq.~\eqref{eq:sm-host-volume} forces
\begin{equation}
 D\geq\frac{1}{2r}\left[
 \frac{\varepsilon N}
 {(1-\gamma^m)(b_0m)(b_0m-1)}-1
 \right].
 \label{eq:sm-path-lower}
\end{equation}
For fixed $m,r,\varepsilon$, this is $\Omega(N)$. Conversely, a range-one classical device can route every sparse question to a designated center, compute an accepting answer, and route the answers back in depth $O(N)$. The separation is therefore diameter optimal for this task.

For any desired fixed gap $1-\eta$, choose a constant
\begin{equation}
 m\geq\left\lceil\log_{1/\gamma}(2/\eta)\right\rceil.
\end{equation}
Then $\gamma^m\leq\eta/2$, while the collision term in Eq.~\eqref{eq:sm-host-volume} is $o(1)$ for every $D=o(N)$. Thus sufficiently large path instances have quantum success one and classical success below $\eta$.

For a $q$-dimensional grid $H=[L]^q$, $N=L^q$ and
\begin{equation}
 \beta_H(R)\leq(2R+1)^q.
\end{equation}
At fixed $m,q,r$, a constant improvement over $\gamma^m$ requires
\begin{equation}
 D=\Omega(L)=\Omega(N^{1/q}).
\end{equation}
Coordinate-by-coordinate routing solves the task in $O(qL)$ range-one depth, again matching the graph diameter up to constants.

\section{Base games}
\label{sec:sm-games}

\subsection{GHZ game}

The referee samples uniformly from
\begin{equation}
 (x,y,z)\in\{000,011,101,110\}.
\end{equation}
Alice, Bob, and Charlie output bits $(a,b,c)$ and win when
\begin{equation}
 a\oplus b\oplus c=x\lor y\lor z.
 \label{eq:sm-ghz-rule}
\end{equation}

\begin{proposition}
\label{prop:sm-ghz}
The GHZ game has $(\cval,\qval)=(3/4,1)$.
\end{proposition}

\begin{proof}
For a deterministic classical strategy, denote the two possible outputs of the three players by $a_0,a_1,b_0,b_1,c_0,c_1$. Winning all four questions would require
\begin{align}
 a_0\oplus b_0\oplus c_0&=0,\\
 a_0\oplus b_1\oplus c_1&=1,\\
 a_1\oplus b_0\oplus c_1&=1,\\
 a_1\oplus b_1\oplus c_0&=1.
\end{align}
XORing the equations gives zero on the left because every variable appears twice, but one on the right. Hence at least one of four equally likely questions is lost. A strategy losing exactly one question exists, so $\cval=3/4$.

For the quantum strategy, share
\begin{equation}
 \ket{\mathrm{GHZ}}=\frac{\ket{000}+\ket{111}}{\sqrt2}.
\end{equation}
On question zero measure Pauli $X$; on question one measure Pauli $Y$. Report eigenvalue $+1$ as bit zero and $-1$ as bit one. The identities
\begin{align}
 XXX\ket{\mathrm{GHZ}}&=+\ket{\mathrm{GHZ}},\\
 XYY\ket{\mathrm{GHZ}}&=-\ket{\mathrm{GHZ}},\\
 YXY\ket{\mathrm{GHZ}}&=-\ket{\mathrm{GHZ}},\\
 YYX\ket{\mathrm{GHZ}}&=-\ket{\mathrm{GHZ}}
\end{align}
produce exactly the required parities in Eq.~\eqref{eq:sm-ghz-rule}.
\end{proof}

For the compiled $m$-copy GHZ task,
\begin{equation}
 p_{\rm c}\leq\left(\frac34\right)^m
 +\left[1-\left(\frac34\right)^m\right]
 \min\!\left\{1,
 \frac{3mK^D}{N}
 \right\}.
 \label{eq:sm-ghz-bound}
\end{equation}

\subsection{Magic-square game}

Alice receives a row $r\in\{1,2,3\}$ and outputs three signs with product $+1$. Bob receives a column $c\in\{1,2,3\}$ and outputs three signs whose product is $+1$ for columns one and two and $-1$ for column three. They win when their signs agree at the intersection.

\begin{proposition}
The Magic-Square game has $(\cval,\qval)=(8/9,1)$.
\end{proposition}

\begin{proof}
If a deterministic strategy won all nine input pairs, its intersection answers would define one $3\times3$ sign array. The product of all row products would be $+1$, while the product of all column products would be $-1$, a contradiction. A strategy that fails at only one intersection exists, so $\cval=8/9$.

For the quantum strategy, share two Bell pairs and use the commuting two-qubit observables
\begin{equation}
\begin{array}{ccc}
 X\otimes\id&\id\otimes X&X\otimes X\\
 \id\otimes Z&Z\otimes\id&Z\otimes Z\\
 X\otimes Z&Z\otimes X&Y\otimes Y
\end{array}.
\end{equation}
Every row product is $+\id$, the first two column products are $+\id$, and the third is $-\id$. The maximally entangled transpose identity
\begin{equation}
 (M\otimes\id)\ket{\Phi_4}
 =(\id\otimes M^{\mathsf T})\ket{\Phi_4}
\end{equation}
ensures agreement at the intersection because all nine displayed observables are equal to their transpose. Thus every input pair is won.
\end{proof}

If each valid three-sign answer is encoded by two independent bits, one may take $a=2$ rather than three. The simplified compiled bound is then
\begin{equation}
 p_{\rm c}\leq\left(\frac89\right)^m
 +\frac{4mK^D}{N}.
\end{equation}
For the finite instances considered here, GHZ gives lower classical thresholds and uses one qubit per player.

\subsection{\texorpdfstring{The $5\times9$ qutrit game}{The 5 x 9 qutrit game}}

For the qutrit game of Sec.~\ref{sec:sm-contextuality}, $b_0=2$, $\gamma=44/45$, and each of three possible outputs can be encoded in two bits. Therefore
\begin{equation}
 p_{\rm c}\leq\left(\frac{44}{45}\right)^m
 +\frac{4mK^D}{N}.
 \label{eq:sm-qutrit-bound}
\end{equation}
The value $44/45$ requires more repetitions than GHZ to reduce the first term to the same value. This game gives an explicit connection to qutrit contextuality, but requires more resources than the GHZ construction.

\section{Experimental protocol and finite-data certification}
\label{sec:sm-experiment}

\subsection{Logical architecture}

Choose a bounded-degree rooted tree with $N$ vertices. A balanced binary tree has maximum undirected degree three and path length $O(\log N)$, reducing accumulated physical error even though the ideal logical circuit depth is constant. {Because the classical bound depends on the host only through $N$, a finite experiment may instead choose the host to minimize the number of teleportation hops; see Sec.~\ref{sec:sm-host}.} {Here ``terminal'' means a candidate teleportation endpoint: every vertex, including the root and internal vertices, is eligible. A vertex denotes a site containing distinct registers for its incident edges, rather than a single qubit.} For the $m$-copy GHZ protocol there are $b=3m$ players, each with a separate copy of the graph. Every such graph contains
\begin{equation}
 1+2|E|=2N-1
\end{equation}
qubits: one game qubit at the root and two half-edge qubits for each of the $N-1$ edges. The total is
\begin{equation}
 Q=3m(2N-1).
 \label{eq:sm-register-count}
\end{equation}
Readout ancillas may be reused locally and are not included.

Each trial proceeds as follows.
\begin{enumerate}
\item \emph{Product initialization.} Initialize all qubits in computational-basis states.
\item \emph{Terminal draw.} {The verifier samples all $3m$ terminals independently and uniformly and sends the corresponding path flags to the device. The game questions have not been revealed.}
\item \emph{Parallel preparation.} {Prepare $m$ GHZ states on the root game qubits and a Bell pair on every edge of every selected path. Preparing a Bell pair on every player-labeled edge is also allowed but unnecessary: the path flags are already available, a classically controlled preparation has the same depth, and unused qubits remain in $\ket0$ (Remark~\ref{rem:sm-selected}).} Bell-pair preparations and distinct GHZ-copy gadgets act on disjoint registers and run in parallel; every fixed GHZ gadget has $O(1)$ internal depth and all gates have bounded fan-in.
\item \emph{Round one.} Perform the selected Bell measurements simultaneously and return every outcome used to determine the Pauli frames.
\item \emph{Question reveal.} {After the round-one transcript is fixed, reveal independent GHZ questions for each copy. The verifier may generate them at any time; what matters is that the device has no access to them before this step (Sec.~\ref{sec:sm-soundness}).}
\item \emph{Round two.} Send the question and frame to each terminal, perform the $X$ or $Y$ measurement selected by the question, and return one bit per player; {the frame is applied to the record as in Eq.~\eqref{eq:sm-frame-bit}.}
\item \emph{Scoring.} Count the trial as a win only if all $m$ copies satisfy the GHZ parity rule. Every missing outcome, loss, timeout, or invalid output counts as a loss.
\end{enumerate}
The last convention avoids a fair-sampling loophole. Postselection on successful quantum events would require applying the same declared postselection rule to the classical benchmark and rederiving the soundness bound.

\subsection{Finite parameter choices}

{For fan-in $K=2$ and online depth $D=1,2,3$, Eq.~\eqref{eq:sm-ghz-bound} gives the values in Table~\ref{tab:sm-finite}. Because the collision term is $3m2^D/N$, doubling $N$ at fixed $m$ preserves the same classical bound while increasing $D$ by one. An entry of $1$ means that the instance does not constrain that depth at all.} The column $q_{\rm iid}$ is the per-copy success that would make $q_{\rm iid}^m=B_{\rm c}$ {at $D=1$} under an independent, identical-copy diagnostic model. It is not an assumption of the certification test, which uses the directly measured all-copy success.

\begin{table}[h]
\caption{{Staged GHZ experiments against fan-in-two classical circuits of online depth $D$. $B_{\rm c}$ is the right-hand side of Eq.~\eqref{eq:sm-ghz-bound}; $q_{\rm iid}$ refers to $D=1$.}}
\label{tab:sm-finite}
\centering
\begin{ruledtabular}
\begin{tabular}{rrrrcccc}
$m$&$N$&$b$&$Q$&$B_{\rm c}$ ($D=1$)&$B_{\rm c}$ ($D=2$)&$B_{\rm c}$ ($D=3$)&$q_{\rm iid}$\\
\hline
1&17&3&99&0.838235&0.926471&1&0.838235\\
1&64&3&381&0.773438&0.796875&0.843750&0.773438\\
1&128&3&765&0.761719&0.773438&0.796875&0.761719\\
2&128&6&1530&0.603516&0.644531&0.726562&0.776863\\
2&256&6&3066&0.583008&0.603516&0.644531&0.763549\\
2&512&6&6138&0.572754&0.583008&0.603516&0.756805\\
2&2048&6&24570&0.565063&0.567627&0.572754&0.751707\\
3&2048&9&36855&0.426956&0.432037&0.442200&0.752999\\
\end{tabular}
\end{ruledtabular}
\end{table}

Increasing $N$ suppresses lightcone collisions but increases the installed register count. Preparing only the selected paths reduces the number of Bell pairs used in a trial. Increasing $m$ suppresses the base classical value exponentially while increasing the collision prefactor linearly. For $m=2$ and $N=512$, the thresholds are $0.572754$, $0.583008$, and $0.603516$ for $D=1,2,3$, respectively. Certification requires the lower confidence bound on the measured all-copy success to exceed the chosen threshold.

\subsection{Statistical confidence bounds}

Let $W_1,\ldots,W_M\in\{0,1\}$ be trial outcomes and let $p_{\rm exp}$ be the true winning probability for identically distributed trials. Define
\begin{equation}
 \widehat p=\frac{1}{M}\sum_{i=1}^{M}W_i.
\end{equation}
For {independent, identically distributed trials with a predeclared sample size,} Hoeffding's inequality gives
\begin{equation}
 \Prb\!\left[p_{\rm exp}<\widehat p-\epsilon\right]
 \leq e^{-2M\epsilon^2}.
\end{equation}
At confidence $1-\alpha$, use
\begin{equation}
 L_\alpha=\widehat p-
 \sqrt{\frac{\ln(1/\alpha)}{2M}}.
 \label{eq:sm-confidence}
\end{equation}
The experiment certifies quantum advantage against the declared classical class when
\begin{equation}
 L_\alpha>B_{\rm c}(m,N,D,K,a).
 \label{eq:sm-certification}
\end{equation}
For $\alpha=0.01$, margins $\epsilon=0.05,0.03,0.02$ require respectively
\begin{equation}
 M\geq922,\qquad2559,\qquad5757
\end{equation}
trials after rounding upward. A one-sided exact binomial interval may replace Eq.~\eqref{eq:sm-confidence} without changing the criterion in Eq.~\eqref{eq:sm-certification}.

{If drifts create dependence between trials, use a predeclared concentration bound valid for that dependence. A martingale bound applies when the conditional classical success in every trial, given the previous record, remains at most $B_{\rm c}$. This requires fresh independent terminals and questions and the stated circuit restriction on every trial. Refreshing the settings alone does not make the trials independent.}

\subsection{Implementable measurement corrections}

For qubits, every path frame is a Pauli $X^AZ^B$ up to phase. Conjugating $X$ or $Y$ by a Pauli changes only the sign of the measured observable. Thus the GHZ terminal need not synthesize a new measurement basis for every frame; it measures $X$ or $Y$ and flips the reported bit when dictated by the frame. {If $r$ is the raw outcome bit and $x=0,1$ selects the $X,Y$ measurement, respectively, then
\begin{equation}
 r_{\rm corr}=r\oplus B\oplus Ax.
 \label{eq:sm-frame-bit}
\end{equation}}
The verifier can compute $(A,B)$ by XORing the Bell labels along the path. The GHZ protocol therefore needs no frame-dependent quantum operation: the verifier may relabel the outcome after the trial. This relabeling uses only the player's own question and the pre-question transcript, so it preserves the locality used in the classical bound. The setting $x_j$ must reach terminal $T_j$ only after the round-one outcomes have been recorded.

For qutrits, the frame lies in the qutrit Weyl group. The $5\times9$ game therefore requires a finite menu of frame-conjugated qutrit basis measurements. A platform with programmable high-dimensional projective measurements can precompile all of them. The number of settings is independent of $N$ and $m$.

\subsection{Scope of certification}
\label{sec:sm-certification-scope}

The experiment certifies a circuit-complexity separation. It does not exclude unrestricted classical algorithms: Proposition~\ref{prop:sm-upper} gives a classical solution of depth $O(\log N)$. The classical competitor is allowed arbitrary precomputation and global nonuniform wiring, but its response after the fresh second-round questions is restricted to depth $D$ and fan-in $K$. A laboratory implementation should therefore state the abstract gate model and, if a timing interpretation is desired, the hardware clock-cycle convention used to translate a response deadline into $D$.

{For GHZ, treating a question-controlled $X/Y$ measurement as an elementary operation gives one measurement layer in round two, before processing the recorded frame correction. With $S=\operatorname{diag}(1,i)$, a selected terminal instead implements the measurement by applying $S^\dagger$ when $x=1$, then $H$, then computational-basis readout: two single-qubit unitary layers followed by measurement. Local terminal/null selection and any device-side correction must also be counted when reporting the executed circuit depth. The depth-one and depth-two classical thresholds refer to the declared response class and do not by themselves establish equal elementary-layer depth or a wall-clock speedup.}

The classical bound is proved analytically, so the experiment measures the quantum winning probability and must implement the stated shallow quantum circuit. In the 99-qubit GHZ test, keeping the three game qubits at fixed physical locations and routing the questions to them classically could exceed $57/68$ without demonstrating the claimed shallow-circuit separation. The executed circuit must be reported with the data, including the distinct terminal registers, product-state initialization, depth of each round, and single-qubit measurements in round two.

The verifier is trusted to sample terminals and questions independently and {to withhold the questions until round two}. The test is not device independent with respect to the declared circuit model. For the specified circuit model, the classical bound is unconditional and explicit at finite size.

\section{A 99-qubit GHZ experiment}
\label{sec:sm-99qubit}

We specify a 99-qubit implementation using one copy of the three-player GHZ game and $N=17$ possible terminals per player. The protocol consists of preparing a GHZ state,
teleporting its three subsystems to independently chosen random locations, {revealing
the GHZ measurement settings to the device} only after the teleportation transcript has been
recorded, and testing whether the observed winning probability exceeds the
classical bound implied by the circuit model.

This instance tests the two-round protocol without disjoint-player amplification. The $m\geq2$ instances in Sec.~\ref{sec:sm-experiment} reduce the base classical value from $3/4$ to $(3/4)^m$, with larger register counts.
{Here 99 counts unencoded data and resource qubits; platform-specific readout ancillas, routing overhead, and error-correction qubits are excluded.}

\subsection{GHZ questions and acceptance rule}

We use one copy of the GHZ game, so
\begin{equation}
m=1,\qquad b_0=3,\qquad a=1,\qquad
\gamma=\frac34 .
\end{equation}
The referee samples uniformly from
\begin{equation}
(x_A,x_B,x_C)\in
\{000,011,101,110\},
\end{equation}
and the three output bits $(a,b,c)$ must satisfy
\begin{equation}
a\oplus b\oplus c=x_A\lor x_B\lor x_C.
\label{eq:sm-99-ghz-rule}
\end{equation}
The ideal quantum strategy uses
\begin{equation}
\ket{\mathrm{GHZ}}
=
\frac{\ket{000}+\ket{111}}{\sqrt2},
\end{equation}
with an $X$ measurement for question $0$ and a $Y$ measurement for question $1$.

The three game qubits are not measured at fixed locations. Instead, each one is
teleported to a uniformly random terminal in its own 17-vertex tree. The game questions are {revealed to the device} only after all round-one teleportation
outcomes have been returned.

\subsection{Qubit layout}

For each of the three GHZ players, use an independent rooted tree with
\begin{equation}
N=17
\end{equation}
vertices and $N-1=16$ edges. One game qubit is placed at the root, and
one Bell-pair register pair is associated with every tree edge. Thus the number of qubits required per player layer is
\begin{equation}
1+2(N-1)=33.
\end{equation}
The complete experiment therefore requires
\begin{equation}
Q
=
3(2N-1)
=
99
\label{eq:sm-99-count}
\end{equation}
qubits.

Equivalently, the hardware contains
\begin{itemize}
\item three game qubits, jointly prepared in a GHZ state;
\item $3\times16=48$ edge register pairs distributed over the tree edges, {of which only those on the selected paths need to be prepared as Bell pairs in a given trial};
\item 96 edge qubits and three GHZ qubits in total.
\end{itemize}

The three trees are logical layers. They need not correspond to three physically
separate processors. What matters for the theorem is that the registers assigned
to different players are distinct.

{If Bell pairs are prepared only on the selected paths, the number of qubits that leave $\ket0$ in a trial is $3+2\sum_j\ell(T_j)$, about $18$ on average for a balanced binary tree, although all $99$ registers must exist so that every terminal is a distinct physical location.}

\subsection{Host graphs}
\label{sec:sm-host}

{The classical bound~\eqref{eq:sm-amplified-exact} depends on the host only through $N$, and the classical restriction applies to round two only, so the host graph may be chosen for the quantum hardware. In the illustrative noise model of Sec.~\ref{sec:sm-99-noise}, attenuation is determined by the number of teleportation hops; general idling and cross-talk errors may also depend on the layout. Three choices for $N=17$:
\begin{itemize}
\item a balanced binary tree (level populations $1,2,4,8,2$; maximum degree three) has mean root-to-terminal distance $\bar\ell=42/17\approx2.5$, that is, about $7.4$ teleportations per trial;
\item a depth-two tree with root degree four and three children per internal vertex ($1+4+12$ vertices; maximum degree four) has $\bar\ell=28/17\approx1.6$, about $4.9$ teleportations per trial;
\item a star (root joined to all sixteen other vertices) has one hop for every non-root terminal, $\bar\ell=16/17$, about $2.8$ teleportations per trial. The circuit must select which root register is Bell-measured with the game qubit. Its growing degree removes the bounded-degree guarantee of constant depth for that selection. The round-two classical bound still applies, but this host does not establish the constant-depth claim for both quantum rounds as $N$ grows.
\end{itemize}
The bounded-degree hypothesis of Theorem~\ref{thm:sm-amplified} is needed only for the constant-depth statement about round one. Using a star reduces the number of teleportation hops, but does not retain the constant-depth guarantee for round one as $N$ grows.}

\subsection{Experimental sequence}

Each experimental trial consists of the following steps.

\paragraph{1. Product-state initialization.}
Reset all 99 qubits to computational-basis product states. No entanglement
is given as an external resource.

\paragraph{2. Random terminal selection.}
{The verifier independently samples
\begin{equation}
T_A,T_B,T_C
\sim {\rm Uniform}\{1,\ldots,17\}
\end{equation}
and sends the corresponding path flags to the device. At this point the GHZ questions $(x_A,x_B,x_C)$ have not been revealed.}

{This ordering is important. The classical bound relies on the questions being independent of everything the device does before round two. The verifier may generate them earlier and keep them private; only the absence of device access is used (Sec.~\ref{sec:sm-soundness}).}

\paragraph{3. Resource preparation.}
Prepare
\begin{equation}
\ket{\mathrm{GHZ}}
=
\frac{\ket{000}+\ket{111}}{\sqrt2}
\end{equation}
on the three root game qubits. In parallel, prepare a Bell pair
\begin{equation}
\ket{\Phi^+}
=
\frac{\ket{00}+\ket{11}}{\sqrt2}
\end{equation}
on {every edge register pair of the three selected paths. Preparing all 48 edges is allowed but unnecessary (Remark~\ref{rem:sm-selected}); unused edge qubits stay in $\ket0$.}

The Bell pairs use distinct qubits and can therefore be prepared in parallel.
The resource preparation is counted as part of the quantum circuit; no
pre-existing long-range entanglement is assumed.

\paragraph{4. Round-one teleportation.}
For each player $j\in\{A,B,C\}$, select the unique root-to-$T_j$ path in the
corresponding tree and perform the Bell measurements required by the path
teleportation protocol.

All Bell measurements along one path act on distinct half-edge registers and may
therefore be carried out simultaneously. Conditioned on their outcomes, {the
joint logical GHZ register is in the state}
\begin{equation}
{
\left(\bigotimes_{j\in\{A,B,C\}}P_j\right)
\rho_{\rm GHZ}
\left(\bigotimes_{j\in\{A,B,C\}}P_j\right)^\dagger,
}
\end{equation}
where, for qubits,
\begin{equation}
P_j=X^{A_j}Z^{B_j},
\qquad A_j,B_j\in\{0,1\}.
\end{equation}
The controller records the Bell outcomes and computes the two-bit Pauli-frame
label $(A_j,B_j)$.

No physical Pauli correction is required at this stage.

\paragraph{5. Question reveal.}
{Only after every round-one Bell outcome has been recorded, reveal one of
\begin{equation}
000,\qquad011,\qquad101,\qquad110,
\end{equation}
drawn with equal probability. The component $x_j$ is sent only to terminal $T_j$.}

\paragraph{6. Round-two measurement.}
At the selected terminal, measure
\begin{equation}
x_j=0 \Longrightarrow X,
\qquad
x_j=1 \Longrightarrow Y.
\end{equation}
The Pauli frame from round one can be absorbed into classical processing. Since
conjugating $X$ or $Y$ by a Pauli operator changes only its sign, the terminal
may perform the usual $X$ or $Y$ measurement and {apply the explicit correction
$r_{j,{\rm corr}}=r_{j,{\rm raw}}\oplus B_j\oplus A_jx_j$ from Eq.~\eqref{eq:sm-frame-bit} to the record.}

{Thus the terminal measurement requires only the usual GHZ $X/Y$ measurement
capability together with delivery of the setting $x_j$ after round one. No quantum feed-forward is needed: the frame correction is a bit flip applied to the record.}

\paragraph{7. Scoring.}
After Pauli-frame correction, let the three output bits be $(a,b,c)$. The trial
is declared a win when Eq.~\eqref{eq:sm-99-ghz-rule} is satisfied.

A missing measurement outcome, unheralded loss, timeout, leakage event that
prevents a valid output, or failed {setting delivery} should be counted as a
loss unless a different treatment was fixed before data collection and the
corresponding classical bound was rederived. In particular, one should not
postselect successful quantum runs and compare them with the unconditional
classical threshold.

\subsection{Classical benchmark}

The classical device may use arbitrary global wiring, shared randomness, and pre-question memory. We give the bounds for fan-in $K=2$ and online depths $D=1$ and $D=2$. For $m=1$, $N=17$, $b_0=3$, $a=1$, and $\gamma=3/4$, Eq.~\eqref{eq:sm-amplified-exact} gives
\begin{align}
B_{\rm c}^{(99)}(D)
&=\frac34+\frac14\min\!\left\{1,\frac{3\,2^D}{17}\right\}
\nonumber\\
&=\begin{cases}
57/68\simeq0.838235, &D=1,\\
63/68\simeq0.926471, &D=2,\\
1, &D\geq3.
\end{cases}
\label{eq:sm-99-bound}
\end{align}
Thus the 99-qubit instance has a nontrivial bound for depths one and two. The threshold exceeds the ordinary GHZ value $3/4$ because the competitor may communicate through global wiring during its bounded-depth response. Table~\ref{tab:sm-finite} gives the corresponding bounds for larger instances.

{The 99-qubit choice is illustrative rather than resource-minimal. For one GHZ copy, this bound is nontrivial whenever $N>3\,2^D$: $N=7$ gives $Q=39$ and threshold $27/28$ at depth one, while $N=13$ gives $Q=75$ and threshold $51/52$ at depth two. These smaller designs require higher winning probabilities.}

The role of the random terminals is precisely to make such coordination unlikely.
One answer bit of a fan-in-two, depth-one circuit has at most two round-two input-wire ancestors in total. Each ancestor belongs to one player and one terminal block; if it belongs to another player, that player's independently chosen terminal hits the block with probability $1/17$. A union bound over the at most two ancestors and three answer bits gives collision probability at most $6/17$, even though the terminals are disclosed in round one.

\subsection{Visibility and error estimates}
\label{sec:sm-99-noise}

{The winning probability of the compiled GHZ test is an exact function of the measured correlators. Let $O_q$ be the parity observable for question tuple $q$ ($XXX$, $XYY$, $YXY$, $YYX$) and $s_q=\pm1$ the sign required by Eq.~\eqref{eq:sm-99-ghz-rule}. Since the questions are uniform,}
\begin{equation}
 p_{\rm exp}=\frac{1+V}{2},\qquad V=\frac14\sum_q s_q\langle O_q\rangle,
 \label{eq:sm-99-visibility}
\end{equation}
{where $\langle O_q\rangle$ is the frame-corrected correlator actually observed, including preparation, teleportation, idling, and readout errors. A failed trial is assigned signed parity $-1$ in the average defining $V$, consistently with counting it as a loss. Exceeding the thresholds in Eq.~\eqref{eq:sm-99-bound} requires $V>23/34\simeq0.6765$ for $D=1$ and $V>29/34\simeq0.8529$ for $D=2$. The illustrative target $p_{\rm exp}=0.96$ corresponds to $V=0.92$.}

A simple multiplicative model gives illustrative component-error targets. Let $V_0$ be the measured signed GHZ visibility in a calibration without teleportation, using the same preparation and readout. For example, a preparation visibility $V_{\rm prep}=0.98$ and independent symmetric readout flips of probability $e=0.005$ on each qubit give $V_0=V_{\rm prep}(1-2e)^3=0.950893\ldots$. We use the rounded illustrative baseline $V_0=0.95$ below.

For completeness, if $F_\pm=\langle\mathrm{GHZ}_\pm|\rho|\mathrm{GHZ}_\pm\rangle$ for the prepared state and $\ket{\mathrm{GHZ}_\pm}=(\ket{000}\pm\ket{111})/\sqrt2$, then
\begin{equation}
 V_{\rm prep}=\Tr\!\left[\rho\,\frac{XXX-XYY-YXY-YYX}{4}\right]=F_+-F_-.
 \label{eq:sm-preparation-visibility}
\end{equation}
Thus $F_+=0.98$ alone implies only $0.96\leq V_{\rm prep}\leq 0.98$, because $0\leq F_-\leq0.02$.

Suppose that every teleportation multiplies each relevant correlator by the same factor $r=1-\epsilon_{\rm t}$ with $0\leq\epsilon_{\rm t}<1$, independent of the question, terminal, and other hops. One teleportation includes a Bell-pair preparation and a Bell measurement. This illustrative model assumes valid outcomes and $0\leq V_0\leq1$. For three independently and uniformly drawn terminals on identical hosts, the exact average is
\begin{equation}
V=V_0\left[\frac1N\sum_{v\in V(H)}r^{\ell(v)}\right]^3
 \geq V_0r^{3\bar\ell},\qquad
 \bar\ell=\frac1N\sum_{v\in V(H)}\ell(v).
 \label{eq:sm-visibility-exact}
\end{equation}
The equality follows by averaging $V_0r^{\ell(T_A)+\ell(T_B)+\ell(T_C)}$ over the three terminals; the inequality is Jensen's inequality. For the binary tree, degree-four tree, and star of Sec.~\ref{sec:sm-host}, the single-terminal factors are, respectively,
\begin{equation}
\frac{1+2r+4r^2+8r^3+2r^4}{17},\qquad
 \frac{1+4r+12r^2}{17},\qquad
 \frac{1+16r}{17}.
 \label{eq:sm-host-attenuation}
\end{equation}

Using that lower bound, $V_0=0.95$ and $(1-\epsilon_{\rm t})^{3\bar\ell}\geq0.92/0.95$ suffice to reach $V\geq0.92$. The corresponding rounded estimates are
\begin{equation}
 \epsilon_{\rm t}\lesssim
 \begin{cases}
 0.4\% & \text{balanced binary tree } (\bar\ell=42/17),\\
0.64\% & \text{depth-two degree-four tree } (\bar\ell=28/17),\\
 1.1\% & \text{star } (\bar\ell=16/17),
 \end{cases}
\end{equation}
{per teleportation. These percentages describe correlator attenuation, not general gate infidelity. The target $p_{\rm exp}=0.94$ ($V=0.88$) permits approximately $1.0\%$, $1.5\%$, and $2.7\%$ attenuation per teleportation. The sample sizes needed for either classical depth are given in Sec.~\ref{sec:sm-99-certification}.}

For a hardware-based estimate, we use the mean single-qubit and two-qubit Pauli error rates $p_1=5\times10^{-4}$ and $p_2=3\times10^{-3}$ and symmetric readout error $e=7\times10^{-3}$ from the GHZ experiment in Ref.~\cite{Kumar2025Contextuality} (Supplementary Fig.~S6). We assume independent depolarizing gate errors, the same readout error for the Bell outcomes, and equal signed preparation-and-terminal-readout correlators $V_0=0.95$ for all four questions. The last assumption is an illustrative calibration, not a measured value from that experiment. Writing each CNOT as a controlled-$Z$ gate between two Hadamard gates on its target, one teleportation uses two controlled-$Z$ gates and six Hadamard gates. For this fixed gate decomposition, with $\lambda_1=1-4p_1/3$, $\lambda_2=1-16p_2/15$, and $s=1-2e$, its attenuation factors are
\begin{equation}
r_X=\lambda_1^4\lambda_2^2s=0.977090\ldots,\qquad r_Y=\lambda_1^6\lambda_2^2s^2=0.962126\ldots.
\label{eq:sm-hardware-hops}
\end{equation}
Propagating $X$ or $Y$ backwards through the teleportation circuit counts four or six single-qubit error locations and one or two Bell-readout bits, respectively. Both observables encounter the two entangling gates. Independent Pauli errors allow the same factors to multiply along paths whose Bell measurements are simultaneous. Let $A_P=N^{-1}\sum_{v\in V(H)}r_P^{\ell(v)}$ for $P=X,Y$. Averaging the four GHZ questions then gives
\begin{equation}
p_{\rm est}^{(m)}=\left[\frac{1}{2}+\frac{V_0}{8}\left(A_X^3+3A_XA_Y^2\right)\right]^m.
\label{eq:sm-hardware-estimate}
\end{equation}
The power $m$ additionally assumes independent copies. For $N=17$ and $m=1$, this gives $0.879$ for the binary tree, $0.908$ for the degree-four tree, and $0.935$ for the star. All three estimates exceed the depth-one bound $57/68$; only the star estimate exceeds the depth-two bound $63/68$. The star has the round-one selection limitation described in Sec.~\ref{sec:sm-host}. These comparisons concern the stated error model and do not establish performance on a particular device. Figure~\ref{fig:main-protocol}(b) uses binary trees filled level by level. These estimates omit additional idling, routing, leakage, and cross-talk errors. They are model projections using reported component rates, not observations of our protocol or predictions for a specified processor layout.

{These numbers guide hardware development only. No fidelity model enters the certification. The measured winning probability includes preparation, teleportation, and readout errors.} For hardware development, it is nevertheless useful to record the success
conditioned on path length and on the different GHZ question tuples. These
quantities diagnose the source of error but should not be used to remove
unsuccessful trials from the certification data.

\subsection{Finite-data certification}
\label{sec:sm-99-certification}

{Let $W_i\in\{0,1\}$ be the outcome of experimental trial $i$, and define the observed winning frequency after $M$ trials by}
\begin{equation}
\widehat p
=
\frac1M\sum_{i=1}^{M}W_i .
\end{equation}
For {independent, identically distributed trials with the number of trials fixed in advance,} a one-sided Hoeffding confidence bound gives
\begin{equation}
L_\alpha
=
\widehat p
-
\sqrt{\frac{\ln(1/\alpha)}{2M}}.
\label{eq:sm-99-confidence}
\end{equation}
At confidence $1-\alpha$, the experiment certifies advantage against the chosen classical depth $D$ whenever
\begin{equation}
L_\alpha>B_{\rm c}^{(99)}(D),\qquad D\in\{1,2\}.
\label{eq:sm-99-certification}
\end{equation}

For $\alpha=0.01$ and an observed winning frequency at least $p_*>B_{\rm c}^{(99)}(D)$, it suffices to choose
\begin{equation}
M>\frac{\ln 100}{2[p_*-B_{\rm c}^{(99)}(D)]^2}.
\end{equation}
For $p_*=0.96,0.94,0.93$, respectively, the resulting integer sample sizes are $156$, $223$, and $274$ against depth one, or $2049$, $12580$, and $184847$ against depth two. These are conditional statements about an observed frequency, not guarantees that a device with mean success $p_*$ will certify at that sample size.

If significant temporal drift or inter-trial correlations are present, the data
should instead be analyzed using a predeclared {concentration bound valid for the observed dependence, under the conditions stated in Sec.~\ref{sec:sm-experiment}}. This changes the statistical analysis but not the
classical threshold in Eq.~\eqref{eq:sm-99-bound}.

\subsection{Calibration sequence}

A practical implementation can be developed in four stages.

\begin{table}[h]
\caption{Calibration stages for the 99-qubit implementation.}
\label{tab:sm-99-stages}
\centering
\begin{ruledtabular}
\begin{tabular}{ccl}
Stage & Approx.\ qubits & Primary objective\\
\hline
A & 3 & GHZ-game preparation and $X/Y$ readout\\
B & {5 or 9} & GHZ correlations after one-hop teleportation\\
C & 20--40 & Random terminals and multi-hop teleportation\\
D & 99 & Full three-tree $N=17$ experiment
\end{tabular}
\end{ruledtabular}
\end{table}

\paragraph{Stage A: GHZ benchmark.}
Prepare and measure the bare three-qubit GHZ game. Measure the winning probability and compare it with the final target. This establishes the
state-preparation and $X/Y$-measurement baseline {and measures $V_0$}.

\paragraph{Stage B: one-hop teleportation.}
Teleport one GHZ subsystem through one Bell pair before performing its $X/Y$
measurement. Verify that Pauli-frame correction restores the expected GHZ parity
for every Bell outcome and every allowed question tuple. {This uses five qubits; teleporting all three subsystems one hop uses nine, and measures $\epsilon_{\rm t}$.}

\paragraph{Stage C: random multi-hop terminals.}
Introduce short trees and randomly choose terminal depths. Compare the corrected
GHZ winning probability as a function of path length. These data estimate how the error depends on path length and help compare the hosts of Sec.~\ref{sec:sm-host}.

\paragraph{Stage D: full experiment.}
Run all three 17-vertex layers simultaneously, independently sample all three
terminals, {reveal} the GHZ questions only after the round-one transcript, and
apply the certification criterion in Eq.~\eqref{eq:sm-99-certification}.

\subsection{Trial data}

For each experimental trial, record
\begin{equation}
\left(
T_A,T_B,T_C,
\mathbf{s}_A,\mathbf{s}_B,\mathbf{s}_C,
x_A,x_B,x_C,
a_{\rm raw},b_{\rm raw},c_{\rm raw},
a,b,c,W
\right),
\end{equation}
where $\mathbf{s}_j$ denotes the complete Bell-measurement transcript for player
$j$, the ``raw'' bits are the terminal measurement outcomes before Pauli-frame
processing, $(a,b,c)$ are the corrected outputs, and $W$ is the final win/loss
bit.

For experimental diagnostics one may additionally record
\begin{itemize}
\item the three path lengths;
\item calibration time and trial timestamp;
\item qubit leakage indicators;
\item reset diagnostics;
\item Bell-pair and Bell-measurement diagnostics;
\item readout-confidence information; and
\item measured {setting-delivery} latency.
\end{itemize}
These auxiliary quantities are useful for understanding hardware performance,
but the primary certification statistic is simply the unconditional frequency of
$W=1$.

\subsection{Certification with 99 qubits}

With the stated implementation assumptions, a lower confidence bound above $57/68$ excludes every randomized classical circuit with arbitrary global wiring, fan-in two, and online depth one. A lower confidence bound above $63/68$ also excludes online depth two.

The circuit-model assumptions, trusted-verifier requirements, and reporting requirements are those of Sec.~\ref{sec:sm-certification-scope}. In particular, the comparison permits global classical wiring and does not impose spatial isolation between players.

The 99-qubit experiment does not yet exploit the main asymptotic amplification
regime. With $m=1$, the intrinsic GHZ classical value remains $3/4$, and the
finite-$N$ collision term raises the classical threshold to
Eq.~\eqref{eq:sm-99-bound}. It specifies the resource preparation, teleportation, delayed questions, and measurements needed for a finite test. A subsequent experiment
with $m\geq2$ uses disjoint-player repetition to replace $3/4$ by
$(3/4)^m$ and thereby produces a substantially lower classical threshold.

\section{Assumptions, resources, and open questions}
\label{sec:sm-limitations}

\subsection{Initialization and resource counts}

For a fixed base strategy, regard each player's local system as one $d$-dimensional register, padding smaller spaces if necessary. The preparation and measurement circuits have constant size for this fixed strategy. For $m$ copies on a degree-$\Delta$ host with $N$ vertices, the number of half-edge qudits is at most
\begin{equation}
 2(b_0m)|E|\leq b_0m\Delta N.
\end{equation}
The total width is $O_{\Game,\Delta}(mN)$. Base-state preparation gates scale as $O(m)$, edge-entanglement gates as $O(mN)$ {if every edge is prepared and as $O(m\log N)$ on a balanced tree if only selected paths are prepared}, and terminal measurements as $O(m)$, but all copies and all edges run in parallel on disjoint registers. {Let $\ell(T_j)$ be the number of edges in $\pi(T_j)$ (equivalently, $\operatorname{dist}_H(r_0,T_j)$ for the chosen shortest-path tree). The number of selected path Bell measurements in a trial is $\sum_{j=1}^{b}\ell(T_j)$, at most $b$ times the host diameter and $O(m\log N)$ on a balanced tree. For the GHZ construction, $b=3m$ and the number of qubits participating in preparation and teleportation in a trial is $b+2\sum_j\ell(T_j)$; the total register count remains $Q=3m(2N-1)$ on a tree, since every terminal must be a distinct physical location.} All entanglement is generated from the product input during round one.

Exact completeness assumes the ideal gate model stated in Theorem~\ref{thm:sm-completeness}: the finitely many preparation and measurement operations of the fixed base strategy are available exactly. For a prescribed gate set that only approximates these operations, the quantum winning probability instead includes the resulting implementation error. The constant circuit depth may depend on the chosen strategy and its local dimension.

For explicit host families such as paths, grids, and bounded-degree trees, a deterministic Turing machine can generate the circuit wiring and gate labels in polynomial time. The quantum family is therefore uniform whenever the fixed base strategy has an explicit finite gate description. The classical lower bound is nonuniform and consequently applies a fortiori to uniform classical families.

\subsection{Frame-dependent measurements}

After teleportation, the endpoint contains $P\rho P^\dagger$. To reproduce the original measurement statistics, the protocol uses the POVM $\{PM_aP^\dagger\}_a$, where the frame $P$ is determined by the Bell transcript. The second round supplies the information needed to choose this measurement. For Pauli measurements, the frame can instead be absorbed into a classical relabeling of outcomes. This removes the need for a physical correction but does not, by itself, give a one-round construction for arbitrary measurements.

\subsection{Noise}

The ideal theorem is noiseless. {If every edge is prepared, a trial prepares $\Theta(mN)$ Bell-pair resources overall; if only the selected paths are prepared (Remark~\ref{rem:sm-selected}), it prepares $\sum_j\ell(T_j)=O(m\log N)$ pairs on a balanced tree. In either case a selected player path of length $L$ uses $\Theta(L)$ Bell measurements and consumes the corresponding links.} Without encoding, constant local noise can therefore drive the success down with $L$. Noise-robust shallow-circuit separations require additional fault-tolerant structure~\cite{BGKT2020}. Provided the circuit and input restrictions are satisfied, the finite experimental criterion requires no noise model: Eq.~\eqref{eq:sm-certification} uses the measured unconditional winning probability. What is absent is an asymptotic theorem guaranteeing that $p_{\rm exp}$ stays above the classical threshold at constant physical noise as $N$ grows.

Balanced bounded-degree trees have paths of length $O(\log N)$, compared with $O(N)$ on a path. At $N=17$, the degree-four tree shortens the paths relative to the binary tree while retaining the bounded-degree guarantee. A star further reduces the number of teleportation hops but requires separate accounting of the selection depth (Sec.~\ref{sec:sm-host}). Heralded link failures are counted as losses under the stated acceptance rule; any different treatment requires a corresponding classical bound.

\subsection{Open questions}

The following questions remain open:
\begin{enumerate}
\item Can one obtain a general one-round construction without restricting the base measurements?
\item Can the contextuality construction be combined with switching-lemma methods to separate the quantum task from stronger classical gate classes such as $\mathsf{AC}^0$?
\item Can an encoded graph-teleportation construction preserve the finite-size advantage under constant local stochastic noise?
\item Can the $5\times9$ qutrit game be reweighted, or replaced by a qutrit perfect game with few inputs and a smaller classical value, to reduce the resources needed for the contextuality-based construction?
\end{enumerate}

\section{Perfect games from contextuality}
\label{sec:sm-contextuality}

The amplification theorem begins with a finite nonlocal game admitting a finite-dimensional perfect quantum strategy. For a finite set of rank-one projectors in dimension $d\geq3$, a finite KS extension provides a way to construct such a game~\cite{Cabello2025BPQS,TrandafirCabello2025}. Figure~\ref{fig:sm-contextuality} shows the inclusion of the initial state and measurement projectors in a KS set, followed by the construction of a perfect bipartite strategy.

In dimension $d\geq3$, a KS set is a finite set of rank-one projectors to which no fixed values $0$ or $1$ can be assigned with at most one value $1$ on any orthogonal pair and exactly one value $1$ on every orthonormal basis contained in the set~\cite{KochenSpecker1967}. A KS set is critical if removing any projector makes such an assignment possible~\cite{Zimba1993}.

Criticality is not needed here. Adjoining a finite KS set in the same dimension gives an extension containing the initial state and measurement projectors, but this inclusion alone does not make the initial contextuality necessary for the perfect game. The added projectors and the bipartite entangled state are additional resources. If the initial projectors already contain a KS subset, that subset can be used directly. The qutrit example below specifies the extension and the resulting game.

{The same extension can be applied to an initially noncontextual collection. Its role here is to supply examples with additional measurements and entanglement, not a quantitative conversion of the original contextuality violation into circuit advantage.}

An infinite-dimensional ambient space causes no difficulty when the example uses a pure state and finitely many rank-one projectors: the state vector and the projector vectors span a finite-dimensional subspace that reproduces the same probabilities. More general contextuality constructions for continuous variables~\cite{Plastino2020PRA} do not by themselves supply a finite-dimensional perfect strategy. Our circuit theorem retains that hypothesis because it must implement the shared state and measurements using fixed-size registers.

\begin{figure}[t]
\centering
\includegraphics[width=0.96\textwidth]{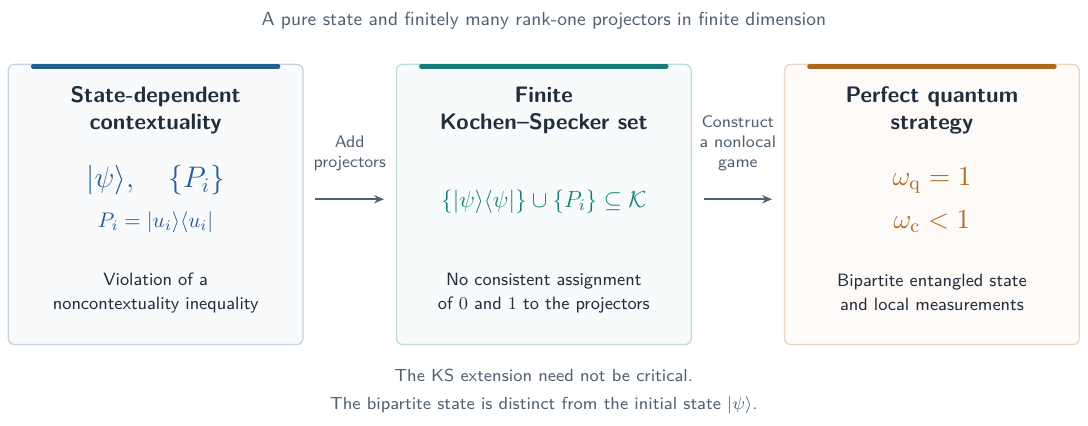}
\caption{Perfect game constructed from a KS extension. A finite collection of rank-one state and measurement projectors in dimension $d\geq3$ can be included in a finite KS set. The resulting orthogonality relations define a nonlocal game with a perfect quantum strategy and classical value strictly below one~\cite{Cabello2025BPQS,TrandafirCabello2025}. An existing KS subset can be used directly. The construction adds projectors and uses a bipartite entangled state distinct from the initial contextuality state; inclusion does not establish that the original contextuality is necessary for the game. No criticality assumption is used.}
\label{fig:sm-contextuality}
\end{figure}

\subsection{Qutrit example: the KCBS inequality} \label{sec:KCBSKS}

We use a qutrit state and five rank-one projective measurements violating the Klyachko--Can--Binicio\u{g}lu--Shumovsky (KCBS) inequality~\cite{Klyachko2008,Araujo2013}:
\begin{equation} \label{eq:kcbs}
    S = \sum_{j=1}^5 P(1,0 \mid j,j+1) \le 2,
\end{equation}
where $P(1,0 \mid j,j+1)$ is the probability of obtaining outcome $1$ when measuring observable $j$ and outcome $0$ for the compatible observable $j+1$, and the sum is taken modulo $5$.

Consider the initial state
\begin{equation} \label{eq:initialstate}
    |\psi\rangle = \tfrac{1}{\sqrt3} 
    \begin{pmatrix}
    1 \\
    1 \\
    -1
\end{pmatrix}
\end{equation}
and the five projectors $\Pi_j=|u_j\rangle\langle u_j|$, each defining a binary measurement with outcomes $1$ for $\Pi_j$ and $0$ for $\id-\Pi_j$, where
\begin{align} 
 \langle u_1 |&=(1,0,0),&\langle u_2|&=(0,1,0),
 &\langle u_3 |&=\tfrac{1}{\sqrt2}(1,0,-1),\nonumber\\[-2mm]
 \langle u_4 |&=\tfrac{1}{\sqrt3}(1,1,1),
 &\langle u_5 |&=\tfrac{1}{\sqrt2}(0,1,-1).
 \label{eq:sm-kcbs-rays}
\end{align}
Adjacent projectors are orthogonal, so $P(1,0\mid j,j+1)=|\langle u_j|\psi\rangle|^2$ and hence

\begin{equation}
    S = 2 + \tfrac{1}{9},
\end{equation}
violating the KCBS inequality \eqref{eq:kcbs}.

For this example, we extend $|\psi\rangle\langle\psi|$ and the five projectors $\{|u_j\rangle\langle u_j|\}_{j=1}^5$ into a specified KS set. All six projectors occur in the 33-ray KS set of Ref.~\cite{Cabello2025KS}; adding the other 27 projectors gives the required extension.

The next step uses the $5\times9$ game of Ref.~\cite{Cabello2025KS}, obtained by distributing the 14 complete bases between Alice and Bob. The bases are assigned using the construction of Ref.~\cite{TrandafirCabello2025}. Alice and Bob share $\ket{\Phi_3}=\tfrac{1}{\sqrt3}\sum_{j=0}^2|jj\rangle$. Alice receives one of five bases and measures its complex conjugate; Bob receives one of nine bases and measures it directly. If their reported ray labels are $u$ and $v$, then
\begin{equation}
P(u,v\mid x,y)=\frac{1}{3}|\langle u|v\rangle|^2.
\end{equation}
The game rejects orthogonal output rays, so this probability is zero for every losing output pair. The complex conjugation is needed because the KS set contains complex rays. See Ref.~\cite{Cabello2025KS} for details.

The resulting perfect strategy uses a specified KS set containing the original state and measurement projectors. The amplification theorem applies to this game without requiring the KS extension or the input assignment to be minimal.

For the uniform distribution over its $45$ input pairs, this game has classical value $44/45$~\cite{Cabello2025KS}. This value is close to one, so many copies are needed to reduce the repeated-game classical value. The example gives an explicit connection to qutrit contextuality; the proposed experiment uses GHZ, whose classical value is $3/4$.

\end{document}